\pdfoutput=1

\documentclass[a4paper, UKenglish, cleveref, autoref, thm-restate]{lipics-v2021}
\nolinenumbers

\newif\ifDraft
\newif\ifArxiv    
\Arxivtrue
\Draftfalse    

\ifDraft
  \newcommand{\nodari}[1]{{\color{red}\#\#{(N:)\footnotesize{ #1 }}\#\#}}
  \newcommand{\peyman}[1]{{\color{blue}\#\#{(P:)\footnotesize{ #1 }}\#\#}}
  \newcommand{\gerth}[1]{{\color{violet}\#\#{(G:)\footnotesize{ #1 }}\#\#}}
\else
  \newcommand{\nodari}[1]{}
  \newcommand{\peyman}[1]{}
  \newcommand{\gerth}[1]{}
\fi

\usepackage[ruled]{algorithm}
\usepackage[noend]{algpseudocode}
\usepackage{bm} 
\usepackage{mathtools}
\usepackage{cite} 
    
\DeclareMathOperator*{\argmax}{argmax}

\DeclareMathOperator{\sign}{\Call{Sign}{}}

\newtheorem{inv}[theorem]{Invariant}

\newcommand{\OhOf}[1]{\ensuremath{O\!\left(#1\right)}}
\newcommand{\OhOfExp}[1]{\ensuremath{O_E\!\left(#1\right)}}

\newcommand{\ThetaOf}[1]{\ensuremath{\Theta\!\left(#1\right)}}
\newcommand{\OmegaOf}[1]{\ensuremath{\Omega\!\left(#1\right)}}
\newcommand{\dOmega}[1]{\ensuremath{\tilde{\Omega}\!\left(#1\right)}}
\newcommand{\dOh}[1]{\ensuremath{\tilde{O}\!\left(#1\right)}}

\newcommand{\U}{\ensuremath{\mathcal{U}}}
\newcommand{\Q}{\ensuremath{\mathcal{Q}}}

\algrenewcommand\algorithmiccomment[1]{\hfill$\triangleright$ \emph{#1}}

\newcommand{\Ell}{\mathcal{L}}
\newcommand{\I}{\mathcal{I}}
\newcommand{\T}[2][H]{T_{#1}\!\left(#2\right)}
\newcommand{\A}{\mathcal{A}}
\newcommand{\R}{\mathbb{R}}
\newcommand{\Sort}{\textup{\texttt{Sort}}}
\newcommand{\Scan}{\textup{\texttt{Scan}}}
\newcommand{\Split}{\textup{\texttt{Distr}}}
\newcommand{\sort}[1]{\ensuremath{\Sort\!\left(#1\right)}}
\newcommand{\scan}[1]{\ensuremath{\Scan\!\left(#1\right)}}

\newcommand{\Ho}{H^\oplus}
\newcommand{\Hbar}{H'}
\newcommand{\bR}{\mathbf{R}}
\newcommand{\tH}{{H_{+}}}
\newcommand{\imax}{{\hat{\iota}}}
\newcommand{\Trand}{{\bar{T}}}

\ifArxiv
  \hideLIPIcs
  \relatedversion{ICALP 2026 conference paper~\cite{AfshaniBrodalSitchinava2026}}
\else
  \relatedversiondetails{Full version}{https://arxiv.org/abs/2605.09464}
\fi

\title{The Impossibility of Simultaneous Time and I/O Optimality for The Planar Maxima and Convex Hull Problems}

\titlerunning{The Planar Maxima and Convex Hull Problems}

\author
  {Peyman Afshani}
  {Aarhus University, Denmark}
  {peyman@cs.au.dk}
  {https://orcid.org/0000-0001-6102-0759}
  {Supported by DFF (Danmarks Frie Forskningsfond) of Danish Council for Independent Research under grant ID 10.46540/3103-00334B.} 
  
\author
  {Gerth Stølting Brodal}
  {Aarhus University, Denmark}
  {gerth@cs.au.dk}
  {https://orcid.org/0000-0001-9054-915X}
  {Supported by Independent Research Fund Denmark grant~9131-00113B.}

\author
  {Nodari Sitchinava}
  {University of Hawai'i at Mānoa, USA}
  {nodari@hawaii.edu}
  {https://orcid.org/0000-0001-8876-4846}
  {Supported by NSF grant 2432018.}

\authorrunning{P. Afshani, G. S. Brodal, N. Sitchinava}

\Copyright{Peyman Afshani, Gerth Stølting Brodal, Nodari Sitchinava}


\category{Track A: Algorithms, Complexity and Games}

\acknowledgements{The authors want to thank Ronitt Rubinfeld for suggesting this topic.}


\EventEditors{Sayan Bhattacharya, Danupon Nanongkai, Michael Benedikt, and Gabriele Puppis}
\EventNoEds{4}
\EventLongTitle{53rd International Colloquium on Automata, Languages, and Programming (ICALP 2026)}
\EventShortTitle{ICALP 2026}
\EventAcronym{ICALP}
\EventYear{2026}
\EventDate{July 7--10, 2026}
\EventLocation{Royal Holloway, University of London, Egham, United Kingdom}
\EventLogo{}
\SeriesVolume{374}
\ArticleNo{115}

\begin{document}

\maketitle

\begin{abstract}
    We prove that no deterministic output-sensitive algorithm for the planar convex hull and maxima problems
can obtain both optimal  time and I/O complexity, where the
optimality is defined with respect to both the input and output sizes.
This explains why the best previous algorithms achieved an optimal I/O bound at the cost of sub-optimal running time (Goodrich et al. [FOCS, 1993]).
To the best of our knowledge, the impossibility of simultaneous optimality was only shown previously for the permutation problem by Brodal and Fagerberg [STOC, 2003].
Our results imply that no optimal deterministic output-sensitive cache-oblivious algorithm exists for either problem. 
In addition, we present simple deterministic algorithms that match our lower bounds 
and that provide a trade-off between time and I/Os. 
On the other hand, a simple modification of our deterministic algorithm results in a randomized algorithm that simultaneously achieves optimal (worst-case) time and optimal expected I/O bounds.

\end{abstract}

\crefname{equation}{Eq.}{Eq.}

\section{Introduction}

We study the planar convex hull problem in the \emph{external memory (EM)}~\cite{EM-model,vitter-io-book} and \emph{cache-oblivious (CO)}~\cite{co-model} settings. 
The surveys on the EM model often mention that 2D convex hull
can easily and optimally be solved in this model via sorting and 
scanning, including the output-sensitive variant~\cite{EM-CH,vitter2001external}.
We report a very surprising discovery 
that the output-sensitive planar convex hull problem cannot be solved optimally both
with respect to the external memory cost and the internal memory computational cost (i.e., running time), 
simultaneously, i.e., by the same algorithm!
Consequently, this means that there is no optimal deterministic cache-oblivious algorithm for the output-sensitive planar convex hull problem. 
We also note that the existing I/O-efficient algorithms for the problem have
substantially sub-optimal (internal memory) computational cost.
We improve this sub-optimality to a very small factor (of an inverse Ackermann-like function) by giving a
very simple algorithm,  but with a non-trivial analysis. 
In addition, we prove lower bounds showing that such sub-optimality is necessary. 
Our lower bound applies even to the simpler problem of computing the maxima of a set of points in 2D~\cite{klp75}, 
a fundamental problem in its own right. See \cref{fig:examples} for examples of the two problems.

To the best of our knowledge, this is only the second time that a ``separation'' between I/O and computational cost has been found, meaning, a problem for which
optimal I/O cost and the optimal computational cost cannot be obtained simultaneously by the same deterministic algorithm but they can be achieved separately by different algorithms.
Such a separation result immediately implies that an optimal cache-oblivious comparison-based
algorithm cannot exist for such problems. This is because an optimal cache-oblivious algorithm performs optimal number of I/Os for all parameters $M$ and $B$, including for $M = \OhOf{1}$ and $B = 1$. However, for these parameters, each I/O can result in at most $\OhOf{M} = \OhOf{1}$ useful (non-redundant) comparisons, i.e., comparisons that reveal new information. Since an optimal algorithm would not perform redundant computation,  the lower bound on suboptimal computation cost (number of comparisons) implies suboptimal I/O complexity for parameters $M = O(1)$ and $B = 1$. 
There are very few instances of such impossibility result for cache-oblivious algorithms~\cite{co-limits,ahz-co-range-reporting,arge-thorup-ram-io-sorting}.
Below, we first present the definitions, as well as the relevant models of computation before presenting a more detailed statement of our results.

\begin{figure}[t]
  \centering
  \includegraphics[height=3cm]{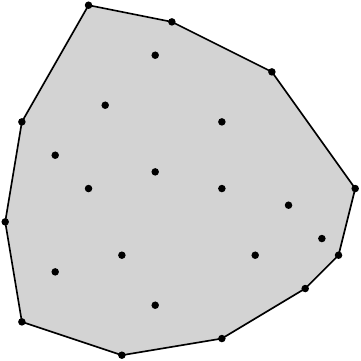}
  \hspace{8em}
  \includegraphics[height=3cm]{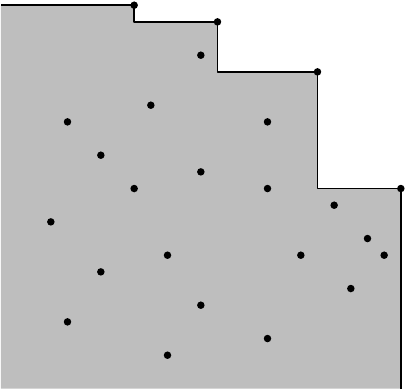}
  \caption{Examples of a convex hull (left) and the maxima (right) of a planar point set.
  We let $N$ and $H$ denote the number of input and output points, respectively.} 
  \label{fig:examples}
\end{figure}

\subparagraph*{Models of computation.} 

The \emph{external memory (EM)} model (a.k.a. the \emph{I/O} model) was  introduced by Aggarwal and Vitter~\cite{EM-model,vitter-io-book} to model the cost of accessing memory in a machine
with two levels of memory: a fast but limited
\emph{internal} memory of size $M$ words and a slow but conceptually unlimited \emph{external} memory. To perform computation, data must reside in the internal memory. Both memories are 
divided into \emph{blocks} of $B$ words and the transfer of data between
the two memories are performed via reading or writing blocks. Each such
transfer defines an \emph{input/output (I/O)} operation and the
\emph{I/O complexity} of an algorithm  is the number of I/Os performed during its execution.
This captures the \emph{cost of memory accesses} and ignores the
computational steps performed in the internal memory (the \emph{running time} in the standard algorithm analysis). 
For convenience we define $m = M/B$ and $n = N/B$.  

However, modern systems consist of more than two levels of memory hierarchy, 
e.g., disk, DRAM, multiple levels of cache, and registers. 
One way to deal with deeper memory hierarchies is via  \emph{cache-oblivious (CO) algorithms}~\cite{co-model,co-model-journal}
which are designed in the classical RAM model (i.e., without using $M$ or $B$ parameters),
but they are analyzed in the \emph{ideal cache} model, a variant of the EM model. 
In particular, during the  analysis, it is assumed that there is a \emph{cache} of $m$ blocks  
and whenever the algorithm accesses an element, a block of size $B$ containing that element is loaded into 
the cache via \emph{an ideal offline paging algorithm}, which makes the optimal choice when deciding which elements to evict
from the cache. 
The ideal paging algorithm can be replaced by 
the \emph{Least Recently Used (LRU)} algorithm as it offers a constant factor bicriteria approximation under a reasonable resource augmentation assumption~\cite{co-model-journal}.
Since the values of $M$ and $B$ are not known to a CO algorithm, if a CO algorithm achieves I/O optimality with respect to an arbitrary choice 
of parameters $M$ and~$B$, it is also optimal for every level of the memory hierarchy.

For many problems, optimal CO and EM algorithms are known.  However,
there are very few actual separation results between EM and CO computations
and we mention them here.  
In the EM model, the problem of sorting a set of $N$ comparable items has the tight
bound of $\ThetaOf{n\log_{m} n}$~\cite{EM-model, vitter-io-book}.
However, Brodal and Fagerberg~\cite{co-limits} showed that no CO comparison-based sorting algorithm can achieve the same
optimal bound unless 
$M = \OmegaOf{B^{1+\varepsilon}}$ for some constant $\varepsilon > 0$. 
This assumption is often known as the \emph{tall cache} assumption.
It is also known that CO data structures for range reporting require asymptotically more
space than their EM counterparts~\cite{ahz-co-range-reporting-conference,ahz-co-range-reporting}. 
Finally, Arge and Thorup~\cite{arge-thorup-ram-io-sorting} specifically studied possible
trade-offs between I/O complexity and time in the word-RAM model, 
highlighting that attaining optimality in both metrics at the same time
can be challenging.

\subparagraph*{Prior results on convex hull computation.}

Given a set $P$ of $N$ points in the plane, the \emph{planar convex hull} problem asks to compute the smallest convex polygon containing $P$.
It can be solved in $\OhOf{N\log N}$ time using many techniques~\cite{andrew-CH,bentley-shamos-CH,graham-CH,preparata-CH,preparata-hong-CH}. 
Graham's scan~\cite{graham-CH} is perhaps the most classical solution which involves sorting and a linear scan using a stack.
The $\OmegaOf{N \log N}$-time lower bound has been proven in various
models~\cite{yao-CH-lb,veb-CH-lb,preparata-hong-CH} and holds even if the
points on the convex hull can be returned in an arbitrary
order~\cite{preparata-hong-CH}. When the size $H$ of the convex hull is small
(e.g., for uniformly random points in a fixed polygon~\cite{dwyer1988convex}),
\emph{output-sensitive} algorithms improve the running time to $\OhOf{N \log H}$~\cite{kirkpatrick-seidel-CH,chan-CH}. Kirkpatrick and
Seidel~\cite{kirkpatrick-seidel-CH} proved a matching $\OmegaOf{N \log
H}$ lower bound in the algebraic decision tree model. 

Obtaining an EM or CO algorithm for the planar convex hull  
is straightforward: sort the points using $\OhOf{n\log_{m}n}$ I/Os via one of the optimal EM or CO sorting algorithms~\cite{EM-model,vitter-io-book,co-model-journal}, followed by the stack-based Graham Scan algorithm~\cite{graham-CH} that requires $\OhOf{n}$ I/Os.

An EM algorithm with the optimal output-sensitive I/O bound for convex hull was  presented by
Goodrich et al.~\cite{EM-CH} and it achieves
$\OhOf{n \log_{m} \frac{H}{B}}$ I/Os but its running time (which was not mentioned) 
is $\OhOf{N\log H + N\log m}$, which is sub-optimal.
Arge and Miltersen~\cite{EM-CH-lb} showed a matching $\OmegaOf{n \log_{m} \frac{H}{B}}$ I/O lower bound 
for the output-sensitive convex hull algorithms.
But surprisingly, there has been no output-sensitive CO algorithm that matches this bound and the best known
algorithm only achieves $\OhOf{n (\log_{m} H + \log\log m)}$ I/Os~\cite{afshani-farzan-CH}.%
\footnote{Note that the tall cache assumption, 
on which the algorithm~\cite{afshani-farzan-CH} relies, 
implies $\log_{m} H = \ThetaOf{\log_{m} \frac{H}{B}}$.}
Interestingly, the extra $\log\log m$ term comes from 
trying to ``guess'' the output size $H$ and, thus, 
the algorithm can be made optimal if $H$ is known.


\begin{table}[t]
  \caption{Previous results and our contributions. 
  $H$ is the size of the convex hull or the number of maxima points, $s$ is an integer parameter, 
  $A_s(\cdot)$ is an Ackermann-like function, $\alpha_s(\cdot)$ is its inverse, $\lambda_s(\cdot)$ is 
      the $s$-th function in its inverse hierarchy (see~\cref{sec:ack} for details), and $\OhOfExp{\cdot}$ denotes an expected complexity bound.
        CO results hold under the tall cache assumption and all results apply to both 2D maxima and 2D convex hull problems.}
  \label{tab:results}

  \renewcommand{\arraystretch}{1.1}
  \tabcolsep=0.5em
  \centering
  \begin{tabular}{ccc@{\hspace{-1.5em}}c} 
    \textbf{Model} & \textbf{Time} & \textbf{I/O Complexity} & \textbf{Notes} \\ \hline
    RAM & $\OhOf{N \log N}$ & - & Classic CH~\cite{andrew-CH,bentley-shamos-CH,graham-CH,preparata-CH,preparata-hong-CH} \\
    RAM & $\ThetaOf{N \log H}$ & - & \cite{kirkpatrick-seidel-CH,chan-CH} \\
    EM  & $\OhOf{N\log N}$ & $\OhOf{n\log_mn}$& Classic I/O \cite{EM-model, co-model-journal, graham-CH} \\
    EM  & $\OhOf{N(\log H + \log m)}$ & $\ThetaOf{n\log_m \frac{H}{B}}$ & \cite{EM-CH,EM-CH-lb} \\
    CO  & $\OhOf{N\log H}$ & $\OhOf{n(\log_m H + \log\log m)}$ & \cite{afshani-farzan-CH} \\
    CO  & $\OhOf{N\log H}$ & $\OhOfExp{n\log_m H }$ & new, randomized \\
    CO  & $\OhOf{N(\log H + \log s)}$ & $\OhOf{n(\log_m sH + \alpha_s(\min\{H, m\}))}$ & new \\
    CO  & $\OhOf{N(\log H + \lambda_s(N))}$ & $\OhOf{n(\log_m H + s)}$ & new \\
    EM/CO & $\OhOf{N\cdot A_s(H)}$ & $\OmegaOf{n\bigl(\alpha_1\bigl(\min\bigl\{H,\sqrt{N/M}\bigr\}\bigr)-s\bigr)}$ & new \\ \hline
  \end{tabular}
\end{table}

\subparagraph*{Our contributions.}

\cref{tab:results} presents a summary of our and previous results.
We develop deterministic CO algorithms for the planar maxima (\cref{sec:maxima-UB}) and convex hull (\cref{sec:cg}) problems 
that can trade-off optimality between time and I/O complexity via a parameter~$s$. 
We prove that no \emph{deterministic} algorithm can obtain both optimal time of $\ThetaOf{N \log H}$
and optimal I/O complexity of
$\ThetaOf{n \log_{m} \frac{H}{B}}$ simultaneously for these two problems (\cref{sec:lowerbound,sec:ch-LB}). Our lower bounds are proven in the comparison-based model for the maxima problem and in a model, where comparisons are generalized to geometric predicates on a constant number of points (formally defined in \cref{sec:ch-LB}), for the convex hull problem.
The lower bound applies to any deterministic algorithm that spends up to $ N \cdot A_s(H)$ time, where $A_s(N)$ is an Ackermann-like function
(formally defined in \cref{sec:ack}) and $s$ is a positive integer parameter.
For example, it applies to all known output-sensitive algorithms which take  $O\left(  N \log H \right) $ time or $O(NH)$ time,
and even to algorithms that take $ O\left( N 2^H \right) $ time, and it gives
an $\OmegaOf{n\cdot \alpha_1\bigl(\min\bigl\{H,\sqrt{N/M}\bigr\}\bigr)}$ I/O lower bound, by setting $s$ 
to a fixed constant value that depends on the constant hidden in the $O(\cdot)$ notation. 
Finally, we show a randomized algorithm that can obtain both optimal worst-case running time and optimal \emph{expected} I/O complexity (\cref{sec:rand}).

Our results imply that there are no optimal deterministic output-sensitive CO algorithms 
or algorithms that have both optimal I/O complexity and optimal running time for these two problems.
The latter was shown only for one other problem before~\cite{co-limits}.
As far as we know, we show the first upper and lower bounds involving inverse
Ackermann-like functions in the area of EM algorithm. 
Finally, our geometric embedding in \cref{sec:ch-LB} gives a technique that
enables us to handle arbitrary geometric predicates involving a constant number of points in 
the context of the classical adversarial argument. 
We find this interesting and believe it could be of independent interest.
For instance, we can get an alternative lower bound of $\OmegaOf{N \log H}$ for the 
output-sensitive convex hull problem.

\subparagraph*{Summary of our approach.}


Our planar maxima algorithm (\cref{sec:maxima-UB}) is a very simple cache-oblivious recursive algorithm that
aggressively prunes the points and  
uses the number of maxima points discovered in recursive
subproblems so far to speed up the subsequent recursions.  
This ultimately leads to a non-trivial recurrence 
(\cref{eq:main-detailed} on page~\pageref{eq:main-detailed}).
\ifArxiv
In \cref{sec:app-main-thm}
\else
In the full version~\cite{full-version} 
\fi
we show that it
solves to an inverse Ackermann-like function. 
A similar approach works for the convex hull problem by replacing the pruning step.
This is done by generalizing the ``bridge finding'' algorithm of 
Kirkpatrick and Seidel~\cite{kirkpatrick-seidel-CH} to an algorithm that can find multiple bridges at the same 
time (\cref{sec:cg}).
However, by randomly switching the order of the recursive calls, we can obtain an algorithm with both
optimal running time and optimal expected I/O complexity (\cref{sec:rand}). 

From a lower bound point of view, the problem is more complicated  and requires more involved techniques.
For the maxima problem, we start in the classical comparison-based adversarial setting where the algorithm can only compare the coordinates
of the points and the adversary decides on the results of the comparisons.
However, we manage to apply this idea in the setting where the input resides in the external memory and only the coordinates of the
points loaded into the internal memory of size $M$ can be compared. 
During this argument, we show that the adversary can give extra information to the algorithm at specific intervals
(\emph{epochs}, defined on page~\pageref{par:epoch})
such that the behavior of the algorithm is simplified: 
among the available ``subproblems'' (captured by the concept of {\em top nodes} in the adversary's binary tree in \cref{sec:lowerbound}) the 
algorithm can choose one subproblem to do one step of the recursion (which creates more subproblems).
The main challenge is to show the lower bound regardless of the choices of the algorithm.
To do that, we develop a potential function analysis and show that the best strategy for an algorithm is to recurse on the 
subproblems with the smallest potential (\cref{sec:mainlb}). 
This analysis leads to the following conclusion: when limited to an amortized budget of $\zeta$ I/Os per element
and a budget of $N\cdot A_s(H)$ total comparisons, 
the algorithm can only ``discover'' $A_{s+2\zeta-1}(1)$ maxima points (\cref{lem:mainlb}). 
This in turn proves the lower bound (\cref{thm:maximalb}).

For the planar convex hull lower bound, we generalize this strategy to work for any geometric predicate:
the algorithm can go beyond comparisons and choose a polynomial~$F$ and a set of $\sigma$ points $p_1, \ldots, p_\sigma$ 
(for some fixed constant~$\sigma$)
and then ask for the sign of the evaluation of $F$ 
on the coordinates of points $p_1, \ldots, p_\sigma$.
To prove a lower bound in this case, we show that the adversarial argument for the maxima problem can be
embedded geometrically in the plane. 
This requires some algebraic geometry ideas: we embed the points close to the $Y=X^\Delta$ curve,  for some fixed constant~$\Delta$,
and designate ``squares'' that represent the potential location of a group of points. 
The squares are conceptually arranged in a tree $T$ of large fan out such that 
for any node $v \in T$ with a square~$Q(v)$, the squares of children of $v$ 
are doubly-exponentially smaller than $Q(v)$ and they are placed equally-spaced inside
$Q(v)$ and centered on the curve $Y=X^\Delta$.
By making the sizes of the square shrink at an appropriate doubly exponential rate, we show that
the geometric predicates (captured by the polynomial $F$) do not provide too much information to the
algorithm (\cref{lem:chresolve}). 
Thus, the lower bound for the maxima problem can be extended to the convex hull problem as well (\cref{thm:chlb}).

\section{Preliminaries}
\label{sec:prelim}

Let $P$ be a set of $N$ comparable elements.
Scanning $P$ while performing $\OhOf{1}$ work per element can be accomplished  in $\ThetaOf{N}$ time and
$\ThetaOf{n}$ I/Os.
Comparison-based sorting of $P$ can be performed cache-obliviously in $\ThetaOf{n\log_{m}n}$ I/Os and 
$\ThetaOf{N\log N}$ comparisons, assuming the cache is tall~\cite{co-model-journal}, i.e., $M\ge B^{1+\varepsilon}$ for some constant $\varepsilon>0$.
We define $\scan{N} = n$ and $\sort{N} = n\log_m n$.

Another fundamental problem is \emph{distribution}, where given a
value $k$, $1 \le k \le N$, the goal is to partition~$P$ into $k$ subsets, $P_1, \dots, P_k$, of
roughly equal size, where each $P_i$ has either $\lfloor \frac{N}{k}\rfloor$ or 
$\lceil \frac{N}{k}\rceil$ elements and all the elements in $P_i$ are larger than or equal to all the
elements in $P_{i-1}$. Each of the resulting subsets~$P_i$ should be stored in contiguous memory.
In the classical comparison-based RAM model, distribution can be solved deterministically in $\ThetaOf{N \log k}$ time,
e.g., by recursive applications of a $\ThetaOf{N}$-time median finding algorithm~\cite{selection}. 
Distribution can be solved cache-obliviously in $\ThetaOf{n\cdot \max\{1, \log_m k\}}$ I/Os and $\ThetaOf{N \log k}$ comparisons~\cite{farzan-MS-thesis,co-model-journal},
assuming a tall cache.
We define $\Split(N,k)=n\cdot \max\{1, \log_m k\}$.

Observe that $\Split(N,k)=n\log_m k$ when $k\ge m$
and $\Split(N,k)=\Scan(N)$ when $k\le m$.
We exploit the following property of the $\Split$ function, which follows from the concavity of logarithms:
If $k=\sum_{i=1}^t k_i$ for some values $k_1, \dots, k_t > m$ and $t \ge m$, then
\begin{align}
  \sum_{i=1}^t \Split\left(N,k_i\right) \le \Split\left(t N,\frac{k}{t}\right) = \Split\left(t N,k\right)-\Split(t N,t).\label{eq:split}
\end{align}
  
\subsection{Ackermann-like Functions and Their Inverses}
\label{sec:ack}

To aid the exposition of the analysis, different papers present different definitions of the recursive functions that are referred to as the \emph{Ackermann} functions~\cite{ackermann-orig,peter-ackermann,robinson-ackermann,sundblad1971ackermann,pettie-ackermann06,pettie-ackermann15,chazelle-ackermann,buck-ackermann,tarjan-ackermann,chazelle-rosenberg-ackermann,fredman-saks-ackermann}. While they differ from the original definition of Ackermann~\cite{ackermann-orig} and are not necessarily equivalent even asymptotically speaking, what they have in common is that they are extremely fast growing and their inverses are extremely slow growing. To aid our exposition, we will use the following functions. 

Let $A_0, A_1, A_2, \dots$ be an infinite sequence of functions, where $A_0(N) = N$, $A_1(N)=2^N$ for any integer $N \ge 1$, and 
$A_{i+1}(N) = A_i^{(N+1)}(N)$, where the notation
$f^{(k)}$  for a function $f$ represents applying $f$ to itself $k$ times,
e.g., $f^{(3)}(N)=f(f(f(N)))$. 

The inverses of $A$ are defined as two distinct functions $\lambda_i(x)$ and $\alpha_N(x)$, where 
$\lambda_i(x)$ is the smallest value $N$ such that $A_i(N) \ge x$ and
$\alpha_N(x)$ is the smallest value $i$ such that $A_i(N) \ge x$.
For example, $\lambda_1(x) = \Theta(\log x)$, $\lambda_2(x) = \Theta(\log^*(x))$ and
in general, $\lambda_i(x)$ is roughly the number of times we need to apply
function $\lambda_{i-1}$ to $x$ to get to a constant; 
therefore, $\lambda_i$ can be thought of as the $i$-th function in the inverse hierarchy.
In contrast, $\alpha_N(x)$ is a much slower growing function:  
 $\alpha_1(x)$ grows slower than any of the
 functions $\lambda_i(x)$ for any fixed $i > 1$. 

While there are many definitions of the Ackermann function, one of the more commonly accepted definitions is 
(the curried version of) the function by P\'eter~\cite{peter-ackermann} and Robinson~\cite{robinson-ackermann}, defined as  
$A_i(N) = A_i^{(N+1)}(1)$, with $A_0(N) = N+1$~\cite{sundblad1971ackermann}. On the other hand, our definition is more similar to the definition of Cormen et al.~\cite{clrs}, albeit with a different base case: they define $A_0(N) = N+1$ and $A_i(N) =  A_i^{(N+1)}(N)$ for all $i \ge 1$, which  implies that $A_1(N) = 2 N+1$ and $A_2(N) = 2^{N+1}(N+1)-1$~\cite[Chapter 19.4]{clrs}. Our definition's base cases $A_0(N) = N$ and $A_1(N) = 2^N$ on the other hand do not have the additional offsets, which are not essential to the demonstration of how fast the Ackermann functions grow, and will make it easier to reason about our upper and lower bounds.  Observe that asymptotically our function is slightly slower growing than that of Cormen et al. but faster growing that that of P\'eter and Robinson. Since our function is not quite the Ackermann function, we will refer to it as \emph{Ackermann-like} function (Cormen et al. also refrain from calling their function Ackermann).

\section{Upper Bound for Planar Maxima}
\label{sec:maxima-UB}

\begin{algorithm*}[t]
  \begin{algorithmic}[1]
    \Procedure{Maxima}{$P, h$} 
    \State \textbf{if} {$|P| \le 1$}  \Comment{Base case} 
      \State \quad Output $P$ and \Return $|P|$
    \State $(P_{2h}, \dots, P_{1}) = \Call{Dist}{P,2h}$\Comment{{Distribute $P$ into $2h$ buckets w.r.t. $X$-coor.}}\label{code:distribute} 
    \State $\Call{Prune}{P_1, \dots, P_{2h}}$ \Comment{For each $P_i$ prune points dominated by points in $\bigcup_{j=1}^{i-1} P_j$}\label{code:prune}
    \State $\Hbar=0$
    \State \textbf{for} {each $P_j \in (P_1, \dots, P_{2h})$}
        \State \quad $H_i = \Call{Maxima}{P_j, h+\Hbar}$
        \State  \quad $\Hbar = \Hbar + H_i$
    \State \Return $\Hbar$
    \EndProcedure
  \end{algorithmic}
  \caption{Algorithm for computing the maxima of a planar point set $P$ with seed $h\ge 1$.} 
  \label{alg:upper}
\end{algorithm*}

Let $P$ be the input set of $N$ points in 2D, listed in an arbitrary order.
The algorithm,  which is presented in \cref{alg:upper}, 
is initially invoked with an integral ``seed''  parameter $h \ge 1$ (not necessarily a constant). 
The choice of the seed
provides a trade-off between the time and I/O complexity of the algorithm. 
At the subsequent recursive invocations, the parameter $h$ will be equal to the initial seed, plus the number of maxima points discovered so far. 

The seed $h$ solves the challenge of not knowing $H$ a priori. 
If we new $H$, we could distribute $P$ into $H$ buckets and it would be easy to show that 
the algorithm would achieve simultaneous optimality.
However, distributing into too many buckets, e.g., $k=H^{\omega(1)}$ buckets, 
results in sub-optimal time of $\omega(N \log H)$.
But if we choose too few buckets, we will have too many recursive levels,
resulting in a sub-optimal I/O bound. 
The total number of discovered output points throughout the computation provides us with a lower bound on~$H$, meaning,
it is always safe to increase the number of buckets as we discover more points and the initial
seed provides us with an initial ``acceleration''. 

We  distribute the points of $P$ into $2h$ buckets of equal size, where $P_1$
contains the rightmost points and $P_{2h}$ contains the leftmost ones.
Next, we remove every point in $P_i$ that is dominated by any of the points in
$P_1, \dots, P_{i-1}$ by a simple scan (see \cref{fig:prune}): 
Process the buckets in order from right to left,
and maintain the maximum $y$-coordinate, $y_{i-1}$, of the points in buckets $P_1, \dots, P_{i-1}$; when processing the next bucket $P_i$, remove
any point in $P_i$ whose $y$-coordinate is smaller than or equal to $y_{i-1}$.
Finally, recurse on each bucket, while increasing $h$ by the number of newly discovered output points in each recursive call. 

\begin{figure}[t]
    \centering
    \includegraphics[scale=0.5]{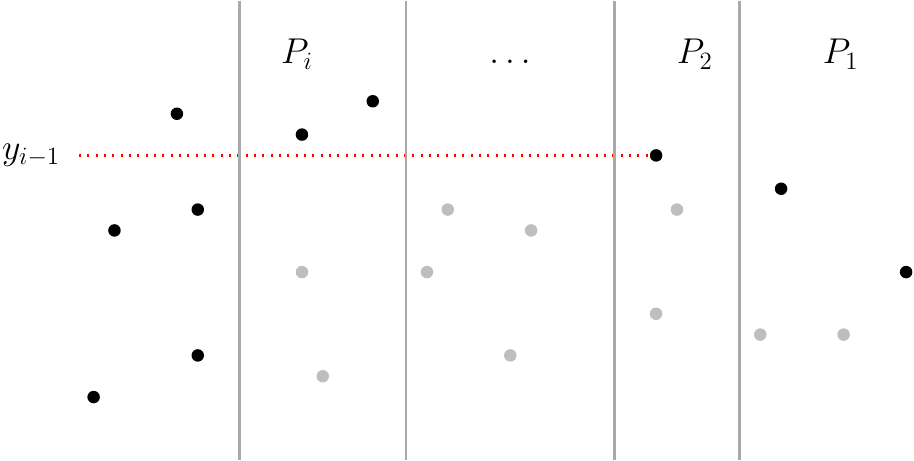}
    \caption{The pruning step: In $P_i$ the points pruned (grey) are those with height at most
    the height $y_{i-1}$ of the highest point to the of right~$P_i$, i.e., the leftmost maximal point in $\bigcup_{j=1}^{i-1} P_j$.}
    \label{fig:prune}
\end{figure}

\begin{restatable}{theorem}{overallIO}\label{thm:overallIO}
For any integer $h \ge 1$,  \textup{$\Call{Maxima}{P,h}$} finds the $H$ maxima points of the input set $P$ of $N$ points in $\OhOf{N \log H + N \log h}$ time and $\OhOf{n \log_m hH + n \cdot \alpha_h(\min\{H,m\})}$ I/Os. 
\end{restatable}

\ifArxiv
We present the proof of \cref{thm:overallIO} in \cref{sec:app-main-thm}. 
\else
\begin{proof}
\ifArxiv
\section{Proof of Theorem~\ref{thm:overallIO}}
\label{sec:app-main-thm}

Here we prove the main upper bound theorem, restated below.

\overallIO*
\fi

Observe that other than the recursive calls, the complexity of each level of
the recursion of the algorithm is dominated by the cost of the distribution
step, which is $\OhOf{\Split(N,2h)} = \OhOf{\Split(N,h)}$.

Let $\T{N, h }$ denote the I/O complexity of $\Call{Maxima}{P, h}$, where $|P| = N$ and $H$ is the number of maxima points in $P$. 
Let $H_i$ denote the number of maxima points of $P_i$ (computed during the recursive call on $P_i$). 
Observe that the pruning step can always be performed with a simple scan and, thus, the distribution cost will 
always dominate the cost of all the other operations at every recursive level. 
\ifArxiv
This means that the I/O complexity of the algorithm can be bound by the following recurrence, where
$c$ is some fixed constant.
\begin{align}
    \T{N, h}  \le \sum_{i = 1}^{2h} \T[H_i]{\frac{N}{2h}, h+\sum_{j=1}^{i-1} H_j} + c \cdot {\Split(N,h)}.   \label{eq:main}
\end{align}

It turns out that to analyze the above recurrence, we first need to identify two base cases. 
This is done by the following lemma.

\begin{restatable}{lemma}{lemrec}\label{lem:recurrence}
\fi
  For any $h \ge 1$, we have the following recurrence:
  \begin{align}
  \T{N, h} \le \left\{
    \begin{array}{ll}
        c_0 \cdot {\Split(N,h)}  & \text{ if } h \ge H, \\
        c_1 \cdot {\Split(N,H)}  & \text{ if } H > h \ge m, \\
        \sum_{i = 1}^{2h} \T[H_i]{\frac{N}{2h}, h+\sum_{j=1}^{i-1} H_j} + c \cdot {\Split(N,h)}   & \text{ if } h < \min(H,m),
  \end{array}
  \right. \label{eq:main-detailed}
  \end{align}
  where $c_0$, $c_1$, and $c$ are some positive constants. 
\ifArxiv
\end{restatable}

Identifying the base cases allows us to focus on the main part of the
analysis which is done in \cref{sec:rec}. 
We will shortly present the proof of \cref{lem:recurrence} 
but we note that 
the lemma can be easily used to bound the time as well.
Notice that 
\else
We show that this recurrence solves to the claimed I/O bound in the full version~\cite{full-version}. We just note that 
\fi
the recurrence that bounds the time 
is exactly the same 
as the one described by \cref{eq:main-detailed},
except the additive term $c\cdot \Split(N,h)$ 
is replaced by  $c N \cdot \log h$ -- the number of comparisons required to perform distribution into $h$ buckets.
But this is equivalent to
setting $m=2$ in the definition of $\Split(N,h)$.
Observe that when $m=2$, the third case 
in \cref{eq:main-detailed}
cannot occur, meaning, the
time can be upper bounded by simply adding up the bounds in the first two branches of \cref{eq:main-detailed}, leading to 
\ifArxiv
the following result:
\begin{theorem}\label{thm:time}
    \textup{$\Call{Maxima}{P, h}$} runs in $\OhOf{N \log h+ N\log H}$ time.
\end{theorem}

In \cref{recurrence-basecase1,recurrence-basecase2}, we first analyze the base cases of~\cref{eq:main} by proving the first two cases in \cref{lem:recurrence}.

\subsection{\texorpdfstring
    {The Base Case $h \ge H$ of Equation~(\ref{eq:main})}
    {The Base Case h ≥ H of Equation~(\ref{eq:main})}
}
\label{recurrence-basecase1}

First, let's consider the case when $h\ge H$. 
In this case, at most $h$ of the subsets $P_{2h},
\dots, P_{1}$ can contain a maxima point, i.e.,
at least $h$ subsets contain no maxima points. 
Observe that during the pruning step, any subset that does not contain any maxima
point will be completely pruned in line~\ref{code:prune} of \cref{alg:upper} for the following simple reason:
Let $p_i$ be the point in $P_i$ with the largest $y$-coordinate, $y_i$, and 
let $Y_{i-1}$ be the largest $y$-coordinate of a point in $P_{i-1}, \dots, P_{1}$.
Observe that if $y_i \ge Y_{i-1}$ then $p_i$ is a maxima point as it has the largest $y$-coordinate
among all the points in $P_{i}, \dots, P_1$ which is a contradiction.
Thus, $p_i$ and, consequently, all points in $P_i$ will be pruned by our pruning step
since their $y$-coordinates lie below $Y_{i-1}$. 

Without loss of generality, we can assume that the last $h$ subsets have no points in them.
In this case, we have the recurrence
\[\T[H]{N,h} \le \sum_{i=1}^{h} \T[H_i]{\frac{N}{2h}, h + \sum_{j=1}^{i-1}H_j} + c \cdot \Split(N,h),
\] 
for some constant $c$. 
Since each $H_i \le H$, i.e., $h$ remains no smaller than each $H_i$ in every
subsequent recursive call, we can use a simple induction to show that
\[
    \T[H]{N,h} \le 2c \cdot \Split(N,h+H) < 4c \cdot \Split(N,h)
\]
and, thus, we can pick $c_0=4c$ in~\cref{lem:recurrence}.

\subsection{\texorpdfstring
    {The Base Case $H > h \ge m$ of Equation~(\ref{eq:main})}
    {The Base Case H > h ≥ m of Equation~(\ref{eq:main})}
}
\label{recurrence-basecase2}

    Next, consider the case when $H > h \ge m$. 
    To ease the presentation, define $\tH_i = \sum_{j=1}^{i-1}H_j$.
    Assume inductively that $\T[H]{N',h'} \le c_1\cdot \Split(N',H)$ for all $N' < N$ and $H > h' > h > m$.   
    Let $\I \subseteq \{1,\dots, 2h\}$ be a subset of indices, such that for each $i \in \I$, $H_i > h+\tH_i \ge m$
    and, thus, for the remaining indices $i \not \in \I$ we have $h + \tH_i \ge H_i$.
    Then 
    \begin{align*}
      \T[H]{N, h} &\le \sum_{i=1}^{2h} \T[H_i]{\frac{N}{2h}, h+\sum_{j=1}^{i-1} H_j} + c\cdot \Split(N,h) \\
       &\le \sum_{i\in\I} \T[H_i]{\frac{N}{2h}, h+\tH_i} + \sum_{i\not\in\I} \T[H_i]{\frac{N}{2h}, h+\tH_i} +c\cdot \Split(N,h) \\
       & \le \sum_{i\in \I} \underbrace{c_1 \cdot \Split\left(\frac{N}{2h}, H_i\right)}_{\text{by Inductive Hypothesis}} + \sum_{i\not\in\I} \underbrace{c_0 \cdot \Split\left(\frac{N}{2h}, h+\tH_i\right)}_{\text{by \cref{recurrence-basecase1}}}+ c \cdot \Split(N,h) \\
       & \le c_1 |\I| \cdot \Split\left(\frac{N}{2h}, H\right) - c_1 |\I| \cdot \Split\left(\frac{N}{2h}, |\I|\right) \tag{\text{by~\cref{eq:split}}} \\
       & \hspace{2em} + c_0 (2h-|\I|)\cdot  \Split\left(\frac{N}{2h}, h+H\right)+ c \cdot \Split(N,h)  \tag{\text{$\tH_i < H$}} \\
       & \le c_1 \cdot  \Split\left(\frac{N|\I|}{2h}, H\right) - c_1 \cdot \Split\left(\frac{N|\I|}{2h}, |\I|\right) \\
       & \hspace{2em} + 2c_0  \cdot \Split\left(\frac{N(2h-|\I|)}{2h}, H\right)+ c \cdot \Split(N,h)   \tag{\text{$h+H < 2H$}} \\
       & = c_1 \cdot  \Split\left(N, H\right) - c_1 \cdot \Split\left(\frac{N|\I|}{2h}, |\I|\right)  \\
       & \hspace{2em} - (c_1-2c_0) \cdot  \Split\left(\frac{N(2h-|\I|)}{2h}, H\right)+ c \cdot \Split(N,h).
    \end{align*}

    Set $c_1 > 2c_0 +2c$ and consider the following two cases:
    \begin{itemize}
      \item Case: $|\I| \le h$
        \begin{align*} \T[H]{N, h} & \le c_1 \cdot \Split\left(N, H\right) - c_1 \cdot \Split\left(\frac{N|\I|}{2h}, |\I|\right) \\
        & \hspace{2em} - (c_1-2c_0) \cdot  \Split\left(\frac{N}{2}, H\right)+ c \cdot \Split(N,h) \\
        & \le c_1 \cdot \Split\left(N, H\right) - c_1 \cdot \Split\left(\frac{N|\I|}{2h}, |\I|\right) \\
        & \hspace{2em} - \frac{(c_1-2c_0 -2c)}{2} \cdot  \Split\left(N, H\right) \\
        & \le c_1 \cdot \Split\left(N, H\right).
        \end{align*}
      \item Case: $|\I| > h$
        \begin{align*} \T[H]{N, h} & \le c_1 \cdot \Split\left(N, H\right) - c_1 \cdot \Split\left(\frac{N}{2}, h\right) \\
        & \hspace{2em} - (c_1-2c_0) \cdot  \Split\left(\frac{N(2h-|\I|)}{2h}, H\right)+ c \cdot \Split(N,h) \\
        & \le c_1 \cdot \Split\left(N, H\right) - \frac{c_1 -2c}{2} \cdot \Split\left(N, h\right) \\
        & \hspace{2em} - (c_1-2c_0) \cdot  \Split\left(\frac{N(2h-|\I|)}{2h}, H\right) \\
        &\le c_1 \cdot \Split\left(N, H\right).
        \end{align*}
    \end{itemize}

\subsection{\texorpdfstring
    {Solving the Main Recurrence ($h < \min\left\{H, m\right\}$) of Equation~(\ref{eq:main})}
    {Solving the Main Recurrence (h < min \{H, m\}) of Equation~(\ref{eq:main})}
}
\label{sec:rec}

In this section we will solve the main part of the recurrence in \cref{lem:recurrence}, i.e., when $h < \min\left\{ H,m \right\}$.
Again, let $\tH_i = \sum_{j=1}^{i-1} H_j$. 

Observe that $h<m$ and $h+H\ge m$ 
happens only once at every level of the recursion; in particular,
consider $T_H(N,h)$ and its recursive calls. 
Consider the first index $i$ such that at the recursive call $T_{H_i}(\frac{N}{2h},h+\tH_i)$
we have $H_i+h+ \tH_i \ge m$ and observe that for the subsequent calls, $T_{H_j}(\frac{N}{2h},h+\tH_j)$, $j>i$,
we will be in one of the base cases, i.e., we will have $h+\tH_j > m$.
As the problem size decreases geometrically by the recursion level,
the total cost of these recursive calls is $\OhOf{\Scan(N)}$.
The remaining recursive calls are when $h+H < m$ which will be the assumption in the rest of the argument. 

We will guess that the recursion in \cref{lem:recurrence} solves to the following:
\begin{align}
    T_H(N,h) \le (c_0+ \gamma(H,h))\cdot \Scan(N) \quad\quad\quad \mbox{when $h+H<m$.}\label{eq:gammaguess}
\end{align}

Note that we have already covered the case when $h \ge H$ in \cref{lem:recurrence} which means we 
can define $\gamma(H,h) = 0$ for when $h\ge H$.
So let us focus on $h > H$.  
Observe that since $h+H<m$, the condition $h+\tH_i < m$ also holds for each recursive subproblem and, therefore, $\Split(x,h+\tH_i) = \Split(x,h) = \Scan(x)$ for any $x>1$. 
Therefore by plugging \cref{eq:gammaguess} for the subproblems we get that for the function $\gamma(\cdot,\cdot)$ to be the solution to the
recurrence in \cref{lem:recurrence} we must have
\[
    \T{N, h} \le \sum_{i=1}^{2h} \left( c_0+c \gamma(H_i,h+\tH_i)\right)\cdot \Scan\left( \frac{N}{2h} \right)+ c\cdot\Scan(N) \le (c_0+ \gamma(H,h))\cdot \Scan(N).
\] 
And since $\Scan\left( \frac{N}{2h} \right) = \frac{\Scan(N)}{2h}$, we get the following recurrence for the function $\gamma$:
\begin{align}
  \gamma(H,h) \ge 1 +  \frac{\sum_{i=1}^{2h} \gamma(H_i,h+\tH_i)}{2h}.\label{eq:grec}
\end{align}

Thus, any function $\gamma(\cdot,\cdot)$  that satisfies the above recurrence can be plugged in \cref{eq:gammaguess} and, thus, it can be used
as an upper bound for $T_H(N,h)$.
Below, we show that this $\gamma$ function is an inverse Ackermann-like function.

We guess that the function is a step function defined as below:
 \begin{align}
     \gamma(H,h) = \left\{
         \begin{array}{ll}
             0 & H\le h \\
             2& h < H \le \beta_1(h) \\
             4& \beta_1(h) < H \le \beta_2(h) \\
             6& \beta_2(h) < H \le \beta_3(h) \\
             \ldots & \ldots \\
             2i & \beta_{i-1}(h) < H \le \beta_i(h) \\
             \ldots & \ldots \\
         \end{array}
       \right. \label{eq:step}
 \end{align}
where $\beta_i$s are a finite number 
of \emph{boundary} parameters that are non-decreasing functions of $h$, 
and increasing functions of $i$.
The idea is that if we can choose the 
$\beta$ functions to satisfy the recurrence of 
\cref{lem:recurrence} then we will have a solution for the recurrence which will be given by the inverses of the
$\beta$ functions.
As we shall see shortly, 
it turns out that we can pick them to be our Ackermann functions.
In particular, picking
$\beta_1(h) = 2^h$, and $\beta_i(h) = A_i(h)$ provide a solution to the recurrence.
We will show this below.

Consider the case when $\beta_k(h) < H \le \beta_{k+1}(h)$ which implies
$\gamma(H,h) = 2k + 2$, for some integer $k\ge 0$.
Let $\gamma_i = \gamma(H_i,h+\tH_i)$.
Observe that for every $i$  we have
\begin{align}
    H_i \le  H \le \beta_{k+1}(h) \le \beta_{k+1}(h+\tH_i), \label{eq:betaub}
\end{align}
which implies 
\[
    \gamma_i = \gamma(H_i,h+\tH_i) \le 2k+2.
\]
Let $\I$ be the subset of indices $\{1, \dots, 2h\}$ such that $H_i \le \beta_k\left( h + \tH_i\right)$, i.e., $\gamma_i \le 2k$. 
We consider two cases: $|\I| \ge h$ and $|\I| < h$. 
We will define $\beta_i$s, such that the second case never happens.

\subparagraph*{Case (1) $|\I| \ge h$.}

Observe that for every index $i \in \I$, we have $\gamma_i \le 2k$ by the definition and 
the other $\gamma$ functions are upper bounded by $2k+2$ by \cref{eq:betaub}.
Thus,
\[
      1 + \frac{\sum_{i=1}^{2h} \gamma_i}{2h} 
  =   1 + \frac{\sum_{i \in \I}^{2h} \gamma_i + \sum_{i \not \in \I}^{2h} \gamma_i}{2h}  
  \le 1 + \frac{ 2k|\I| + (2k+2)(2h-|\I|)}{2h} 
  \le 2k+2 
  =   \gamma(H,h),
\]
which means \cref{eq:grec} holds.
 
\subparagraph*{Case (2) $|\I| < h$.}

Our main idea here is to set up the function $\beta_{k+1}$ such that this case is not possible, i.e., leads to a contradiction
which only leaves the previous case, meaning, the recursion in \cref{eq:grec} holds.
To reach the contradiction, 
consider the first $h+1$ indices, $i_1, \dots, i_{h+1}$, that are not in $\I$ and since $|\I| < h$, they exist.
Define $\tH'_i = \sum_{j=1}^{h+1}H_{i_j}$ and observe that trivially $\tH'_i \le \tH_i$.
Consider an index $i_j$, $1 \le j \le h+1$, and observe that since $i_j \not \in \I$, it follows that 
\begin{align}
  H_{i_j} > \beta_k(h+\tH_{i_j}) \ge \beta_k(h+\tH'_{i_j}) = \beta_k\left(h+\sum_{j'=1}^{j-1} H_{i_{j'}}\right). \label{eq:brec}
\end{align}
Now, we are ready to define the $\beta$ functions. We define $\beta_0(x) = x$.
And  consider two cases: $k = 0$ and $k > 0$. 

Case (a): $k = 0$, i.e, $ h < H \le \beta_1(h)$. Then
observe that we have $\beta_0(h+\tH'_{i_j}) = h + \tH'_{i_j}$. 
\begin{align*}
    H_{i_1} & > \beta_0(h) = h \\ 
    H_{i_2} & > \beta_0(h+H_{i_1}) \ge 2h \\ 
    H_{i_3} & > \beta_0(h+H_{i_1} + H_{i_2}) \ge 4h \\ 
    & ~\vdots  \\
    H_{i_{h+1}} & > 2^{h+1}.
\end{align*}
This implies that $H \ge \tH'_{i_h} > 2^h$.
Thus, if we set $\beta_1(x) = 2^x$, $H > \beta_1(h)$, which contradicts that $h < H \le \beta_1(h)$ for this case. 

Case (b): $k > 0$, i.e.,  $\beta_k(h) < H \le \beta_{k+1}(h)$.
The remaining functions in the hierarchy can be set up in a similar fashion, except that we can be much looser in our
estimation. 
\cref{eq:brec} implies
   $H_{i_j} > \beta_k(\sum_{j'=1}^{j-1} H_{i_{j'}}) \ge \beta_k(H_{i_{j-1}})$. 
Therefore, we can set
   $\beta_{k+1}(h)= \beta^{(h+1)}_k(h)$
where $\beta^{(h)}_k(h)$ refers to applying the $\beta_k$ function $h$ times to itself, which is exactly the way the Ackermann function is defined. 
Thus, $\beta_k$ can be taken as the $k$-th function in our Ackermann hierarchy. And consequently, $\gamma(H,h) = 2\cdot\alpha_h(H)$. 

\else
the claimed time bound.
\fi

\end{proof}
\fi

\cref{thm:overallIO} shows that we get almost optimal bounds, with the trade-off between optimal time and optimal I/Os being defined by the initial seed parameter. In \cref{sec:lowerbound} we prove that no algorithm can do asymptotically better.

\begin{corollary}\label{cor:bounds}
  For any  integer $s \ge 1$, the maxima problem on a planar set of $N$ points 
  can be solved with
  (i) optimal  time of $\OhOf{N (\log H + \log s)}$ and $\OhOf{n(\log_m sH + \alpha_s(\min\{H,m\}))}$ I/O complexity
  cache-obliviously, or 
  (ii) optimal I/O complexity of $\OhOf{n(\log_m H + s)}$ and either 
  $\OhOf{N(\log H + \lambda_s(N))}$ time cache-obliviously or  
  $\OhOf{N(\log H + \lambda_s(m))}$ time cache-aware. 
\end{corollary}

\begin{proof}
The first bound follows from \cref{thm:overallIO} by calling $\Call{Maxima}{P,h}$ with the initial
seed $h = s$. 
The second bound is obtained by calling it with the initial seed $h = \lambda_s(N)$ for a cache-oblivious algorithm, or $h=\lambda_s(m) \le \lambda_s(N)$ if the cache parameters $M$ and $B$ are known.
Observe that we can assume $N\ge m$ because otherwise $N < m \le M$, i.e., the whole input fits in the internal memory. 
The claimed bounds follow because $\alpha_{\lambda_s(N)}(\min\{H,m\}) \le \alpha_{\lambda_s(N)}(N) = s$.
\end{proof}

\section{Lower Bound for Planar Maxima}
\label{sec:lowerbound}

Consider a set $P$ of $N$ points in the plane. 
In this section we show a lower bound for computing the maxima 
of $P$. 
We use the classical adversarial setting based on a rooted binary tree~\cite{DBLP:journals/ipl/BorodinGLY81} for our lower bound. 
While this setting has been used in the past for proving lower bounds for  a single metric of an algorithm (e.g., either for time~\cite{DBLP:journals/ipl/BorodinGLY81} or for I/Os~\cite{em-tree-lb}), the biggest challenge for us is the need to balance the amount of I/Os performed by the
algorithm versus the running time.
 
We prove a lower bound for I/O cost of any deterministic algorithm $\A$ 
that computes the maxima of any set of $N$ points with $H$ maxima points and
uses at most $N\cdot A_{s}(H)$ comparisons, for an integer parameter $s \ge 1$.
To simplify the presentation, we assume that the algorithm uses at most $N A_s(H)$ time, for a parameter $s\ge 3$. 
If the algorithm uses at most $c N A_s(H)$ time, for a fixed constant $c$, then we can bound  $c NA_s(H)  \le N A_{s+c'}(H)$,
for a fixed constant $c'$ depending on $c$,
and this only changes the constants in our lower bound.

We work in the classical comparison-based model.
Comparisons between $X$- and $Y$-coordinates are the only way the algorithm can glean information about the relative position of the points but
the algorithm has unlimited computational resources as well as unbounded capacity to recall all the previous comparisons.
We consider two cost functions:
the total number of comparisons performed by the algorithm, 
which is a lower bound on the time, and the number of I/Os. 
We do not require the ``blocked access'' restriction of the I/O model, meaning, the algorithm
can have the power of ``random access'' by being able to 
read or write any $B$ locations on the external memory via one I/O.
At the end, the algorithm terminates the computation and announces the set of maxima points and we obviously
require the algorithm to be correct. 

\subsection{Resolving Comparisons}
\label{sec:res-comp}

The adversary maintains a perfect binary tree $T$.
Each node $v$ of $T$ is associated with a square region $R_v$ on the plane: 
if $v$ is the $i$-th node at depth $d\ge 0$ (for $i=1, \dots, 2^d$), it is
associated with the square region $R_v = \left(\frac{i-1}{2^d},
\frac{i}{2^d}\right] \times \left(1- \frac{i}{2^d}, 1-\frac{i-1}{2^d}\right]$
(see~\cref{fig:maximasq} for an example). That is, for the root $r$: $R_r = (0,1]
\times (0,1]$; for an arbitrary node $v$, its left and right
children are associated with the upper-left and lower-right quadrants of $R_v$,
respectively. 

\begin{figure}[t]
    \centering
    \includegraphics[scale=.5]{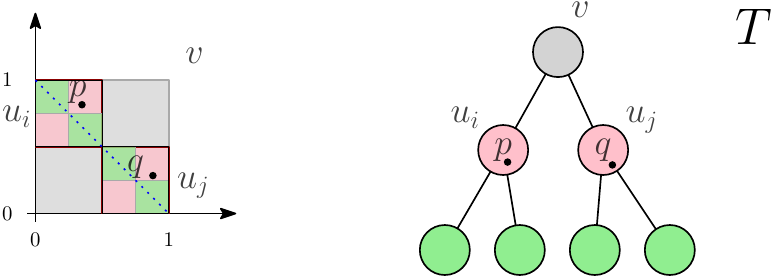}
    \caption{Conceptual regions (on the left) assigned to the nodes of $T$ (on the right). Upon the completion of a comparison between two points $p, q \in v$, their assignment to nodes in the subtrees and, equivalently, to the corresponding regions maintains the consistency with future comparisons.}
    \label{fig:maximasq}
\end{figure}

Throughout the algorithm, the adversary will maintain an assignment of  points to the nodes of $T$
while maintaining the following invariant:

\begin{inv}[Tree invariant]\label{inv:tree}
  If a point $p$ is assigned to a tree node $v$ (denoted $p \in v$), then $p$ can be placed anywhere within $R_v$ consistent with the outcomes of all prior comparisons performed by the algorithm.
\end{inv}

\begin{definition}[Ordered pairs]
A pair of points $p \in v$ and $q \in u$ is called an \emph{unordered pair} if one of $v$ or $u$ is the ancestor of the other one (including $v = u$). Otherwise $p$ and~$q$ is an \emph{ordered pair}.
\end{definition}

\begin{observation}\label{ob:ordered-strategy}
There is only one way to consistently resolve a comparison between an ordered pair of points $p \in u_i$ and $q \in u_j$, because the regions $R_{u_i}$ and $R_{u_j}$ have non-overlapping $X$- and $Y$-ranges.
\end{observation}

Initially, all points are assigned to the root of $T$. Whenever two points are
compared, the adversary produces the outcome of the comparison according to the
following strategy and announces the outcome to the algorithm.
A comparison between an ordered pair  $p$ and $q$ is resolved according to the only consistent way, as per \cref{ob:ordered-strategy}.
The next definition covers how unordered pairs are handled. 

\begin{definition}[Default strategy]\label{strategy:default}
  
  A comparison between a pair of unordered points $p\in v$ and $q \in u$ is resolved as follows. 
    W.l.o.g., let $v$ be an ancestor of $u$ in $T$. 
    If $v=u$, then $p$ is moved
    to the left child of $v$ and $q$ is moved to the right child.
    If $v \not =u$, then $p$ is moved to the child of $v$ that is not the ancestor of $u$.
    In both cases,  $p$ and $q$ become an ordered pair and the comparison is resolved according to~\cref{ob:ordered-strategy}.
\end{definition}

\begin{observation}
  The adversary's default strategy for resolving a comparison between points $p$ and $q$ maintains the tree invariant. 
\end{observation}

\subsection{Adversarial Strategy}

In this subsection, we describe additional definitions and concepts used for our adversarial strategy.
We will use the notation $T(v)$ to refer to the subtree of $T$ rooted at $v$.
We say a \emph{$t$-descendant} of a node $v$ is a node that lies at distance $t$ below $v$ (e.g., $v$ is the $0$-descendant of itself and the children of $v$
are its $1$-descendants). 
Each point $p$, is labeled as either a \emph{deep point} or an \emph{ordinary point}. 
We will explain how this labeling is performed during our adversarial strategy.
Conceptually, deep points represent points that get involved in way too many comparisons and, thus, the vast majority of 
points will be ordinary points. 
We call a non-empty node $v \in T$ a \emph{top node} if all ancestors of $v$ are empty (i.e., contain no points).
The tree invariant implies that the number of top nodes is a lower bound on the number of (current) maxima points.
The \emph{initial size of $v$}, denoted $N_v$, is the number of ordinary points in $v$ the first time $v$ becomes a top node.

\subparagraph*{Charges.}
We will maintain a non-negative integer \emph{charge} with each ordinary point.
Later, we show that the sum of charges across all ordinary points will be a lower bound on the total number of points accessed in the external memory during the execution of algorithm $\A$. This is done via the classical ``amortization'' technique where for every point accessed in the external memory, we transfer $\OmegaOf{1}$ charge to some ordinary point. 
The adversary will maintain the following invariant:

\begin{inv}[Equality of charges]\label{inv:charge}
    For any top node $v$, all ordinary points in $T(v)$ have equal charges.
\end{inv}

Charges will be the way the adversary controls the termination of the algorithm.
In particular, whenever ordinary points at a top node $v$ accumulate $\zeta$ charges, for a parameter $\zeta$ to be chosen later,
the adversary will use the following strategy. 

\begin{definition}[Node termination]\label{def:special} 
  The adversary \emph{terminates} a top node $v$ by
  picking an arbitrary point $p \in v$ and fixing its coordinates to those of the northeast corner of $R_v$. For the remaining points at the nodes in $T(v)$ the adversary fixes their coordinates arbitrarily within the regions of their respective nodes. 
 The coordinates of all points in $T(v)$ are then announced to the algorithm. 
\end{definition}

First, observe that the adversary can perform node termination because by the tree invariant the points can be placed anywhere within the regions of their respective nodes and the relative order within each region is unknown to the algorithm. 
Second, this effectively reduces the size of the maxima among the points in $T(v)$ to 1, essentially pruning all the other points in $T(v)$, because $p$ will dominate all of them.
Finally, terminated nodes are still top nodes, and while the algorithm can still issue comparisons that involve the points in $T(v)$, these points are now ordered (since their 
coordinates are fixed and announced to the algorithm) and, thus, the results of these comparisons are already known.

\subparagraph*{Epochs.}\label{par:epoch}
The adversary operates in \emph{epochs} where during each epoch, 
the adversary resolves comparisons issued by $\A$ via the default strategy. 
Under some conditions, the adversary decides that the current epoch has to end, gives some extra information
to the algorithm and then \emph{transitions} to the next epoch. 
Crucially, the number of top nodes only changes during the transition.
We use $h_i$ to denote the number of top nodes at the start of epoch $i$.
Initially, we start at epoch $1$ when the algorithm $\A$ has issued no comparisons, the root $r$ of $T$ is
the only top node, i.e., $h_1 = 1$, all points are ordinary, with charge zero, and they are placed in $r$.
We will now present the details behind the transition process.

Consider an arbitrary epoch $i$. 
During the epoch the adversary resolves the comparisons issued by the algorithm (via the default strategy) until 
for some top node $v$, the number of ordinary points in $v$ reduces to $N_v/2$ (recall that $v$ started with $N_v$ ordinary points).
If the ordinary points of $v$ had $\zeta-1$ charge, we increase their charge to $\zeta$ and 
terminate $v$ using \cref{def:special} and the epoch $i$ continues.
Otherwise,  epoch $i$ ends and the adversary transitions to epoch $i+1$ by performing the following.
Define the function $d(x)=2^{A_{s}(x)}$ and let $d_i = d(h_i)$.
First, the adversary labels every point that is in a $t$-descendant of $v$ for any $t > d_i$ \emph{deep}.
Next, consider every node $u$ that is a $t$-descendant of $v$ for $1 \le t \le d_i$.
All points in $u$ are moved to an arbitrary $d_i$-descendant of $v$ that is also a descendant of $u$. 
The remaining  $\frac{N_v}{4}$ points of $v$ are then  
distributed equally among all $d_i$-descendants of $v$.
These moves make a number of pairs among the moved points ordered. 
To do this, the adversary provides the information about the relative order of the newly ordered pairs (according to~\cref{ob:ordered-strategy}) to the algorithm for free.
The number of $d_i$-descendants of $v$ is $2^{d_i}$ and,
thus, each $d_i$-descendant will receive 
$\frac{N_v}{4\cdot 2^{d_i}}$ points. 
At this point neither $v$, nor any of its $d$-descendants for $d < d_i$   contain any points. Therefore, the $d_i$-descendants of $v$ become new top nodes and
we call them \emph{activated}.
The charges of all ordinary points in the newly activated top nodes are incremented by one, 
which preserves~\cref{inv:charge}.

If at any point, all nodes are terminated, 
then every top node contains exactly one point that dominates the others in its subtree. Since this information is available to the algorithm, it can announce the result and terminate. 
It is important to note that termination means that each ordinary point in $T(v)$ has received
exactly $\zeta$ charge. 
 
\begin{observation}\label{ob:topnodes}
    In each epoch $i > 1$: $h_{i} = h_{i-1} + 2^{d_{i-1}} - 1 = h_{i-1} + 2^{d(h_{i-1})}-1$.
\end{observation}

\begin{restatable}{lemma}{lemnumpoints}\label{lem:numpoints}
    The initial size of every node $v$ activated at the end of epoch $i-1$ is $N_v \ge \frac{N}{h_{i}^2}$.
\end{restatable}
\begin{proof}
  \ifArxiv
        By induction on $i$. Initially, the root node is the only top node ($h_1 = 1$) and at the start of epoch $i=1$ it contains all $N = \frac{N}{h_{i}^2}$ points (we say root is activated at the end of epoch 0).
    Now consider an arbitrary epoch $i>1$ and
    let $u$ be an ancestor of $v$ whose resolution resulted in the activation of $v$. At the time of $u$'s resolution, at least $N_u/4$ ordinary points of $u$ are distrbuted across $2^{d_{i-1}}$ nodes, one of which is $v$. 
    Let $j-1 \le i-2$ be the epoch at the end of which $u$ was activated with $N_u$ ordinary points. 
    Then, by the inductive hypothesis, $N_u \ge \frac{N}{h_j^2}$. Thus, $v$ received at least $\frac{N_u}{4\cdot 2^{d_{i-1}}} \ge \frac{N}{4h_j^2 \cdot 2^{d_{i-1}}}$ points. 
    Observe that $4h_{i-1}^2 \le 2^{2^{A_{s+1}(h_{i-1})}} = 2^{d(h_{i-1})} = 2^{d_{i-1}}$ for all $h_{i-1} \ge 1$ and $s \ge 1$, and by \cref{ob:topnodes}:
\begin{align}
    4 h_{i-1}^2 \cdot 2^{d_{i-1}} \le  2^{2d_{i-1}} \le (h_{i-1}-1 + 2^{d_{i-1}})^2 = h_{i}^2. \label{eq:hi-progression}
\end{align}

    Then, since $j \le i-1$ and $h_i$s are monotonically increasing, $v$ is activated with at least $\frac{N}{4h_j^2 \cdot 2^{d_{i-1}}} \ge \frac{N}{4h_{i-1}^2 \cdot 2^{d_{i-1}}} \ge \frac{N}{h_{i}^2}$ points.

    \else
    By induction on $i$ (presented in~\cite{full-version}).
    \fi
\end{proof}

\begin{restatable}{lemma}{lemprune}\label{lem:prune}
  The resolution of $v$ at the end of epoch $i$ creates at most $\frac{N_v}{2^{h_i}}$ deep points, where $N_v$ is the initial size
  of $v$ at the time of its activation. 
\end{restatable}
\begin{proof}
    Recall that $\A$ is limited to a budget of $N\cdot A_s(H)$ comparisons. Moreover,
   during each epoch $i$, if algorithm $\A$ performs more than $N \cdot A_s(h_i)$ comparisons, the adversary can use node termination on all top nodes and force $H = h_i$.
   Consequently, $\A$ is limited to a budget of $N\cdot A_s(h_i)$ comparisons in every epoch $i$.
    By \cref{lem:numpoints} and monotonicity of $h_i$s, we have $N \le N_v h_i^2$. Since each comparison moves at most 2 points one level lower,  after $N\cdot A_s(h_i)$ comparisons in epoch $i$, the total number of points that can be deeper than $d_i$ is at most
    $\frac{2 N\cdot A_s(h_i) }{d(h_i)} \le \frac{2 N_v h_i^2 A_s(h_i) }{d(h_i)} = \frac{2 N_v h_i^2 A_s(h_i) }{2^{A_{s}(h_i)}} < \frac{N_v}{2^{h_i}}$,
      where the last inequality follows from the fact that for all integers $s\ge 3$ and $x\ge 1$:
        $2^{A_{s}(x)} > 2x^2 A_s(x) 2^{x}$. 
\end{proof}

Our main technical lemma here is the following.
\begin{lemma}\label{lem:mainlb}
    The number of top nodes is upper bounded by $A_{s+2\zeta-1}(1)$.
\end{lemma}

The proof is postponed to the next subsection but below we quickly show that this is sufficient to give us a lower bound. 

\begin{theorem}\label{thm:maximalb}
  Let $P$ be a planar point set of size $N$ with at most $H$ maxima points, and $s \ge 1$ be an integer parameter. Then any algorithm for computing the maxima of $P$ that uses at most $N\cdot A_s(H)$ comparisons requires
  $\OmegaOf{\frac{N}{B}(\alpha_1(\min\left\{ H, \sqrt{\frac{N}{M}}\right\})-s)}$ I/Os. 
\end{theorem}
\begin{proof}
    We first claim that the adversary creates an instance with at most $H$ maxima points. 
    We choose $\zeta = \frac{\alpha_1(\ell) - s}{2}$ and let $\ell = \min\left\{ H, \sqrt{\frac{N}{4M}}\right\}$. Then by the definition of the $\alpha$ 
    function we have $A_{s+2\zeta-1}(1) < \ell$, i.e., by \cref{lem:mainlb}, the number of top nodes will always
    be smaller than $\ell \le H$ and, thus, the claim holds. 

    Let $V'$ be the set of nodes of $T$ that have been resolved during the execution of $\A$ and let $Z = \sum\limits_{v \in V'} Z_v$ be the total charge across all ordinary points by the end of $\A$, where  $Z_v$ is the increase in charges due to resolution of each $v\in V'$. 
    Observe that by summing the bound given in \cref{lem:prune}, we get that at most a constant fraction of the points can be labeled deep, meaning,
    most of the points will be ordinary. Then, 
    since the adversarial strategy ensures that each ordinary point receives a charge of $\zeta$, we get $Z=\OmegaOf{N\zeta}$.
    We now show that at least $Z/4$ points must have been accessed in the external memory to resolve comparisons performed by $\A$, implying $\frac{Z}{4B} = \OmegaOf{\frac{N \zeta}{B}} = \OmegaOf{\frac{N}{B}(\alpha_1(\ell) - s)}$ I/O lower bound. 

    Observe that in the final epoch $j$, the number of top nodes $h_j \le \ell \le \sqrt{\frac{N}{4M}}$. Since $j$ is the final epoch, every top node $v$ must have been activated prior to the end of some epoch $j'-1 \le j$ and, by \cref{lem:numpoints} and monotonicity of $h_i$s, starts with $N_v \ge \frac{N}{h_{j'}^2} \ge \frac{N}{h_j^2} \ge 4M$ points. 
    By the time $v$ is resolved, at least $N_v/2 \ge 2M$ points of $v$ have been moved to the lower nodes due to comparisons. Thus, all these points could not be kept in the internal memory from activation till resolution of $v$, i.e., at least $N_v/2 - M \ge N_v/4$ of them must have been accessed in the external memory during this period. 
    But we also know that at the end of the resolution of $v$, each ordinary point in $T(v)$ receives an additional charge, i.e., the overall charge in $T$ is increased by $Z_v \le N_v$. 
    Summing over all nodes that have been resolved during the execution of $\A$, we get that at least $\sum_{v \in V'} \frac{N_v}{4} \ge \sum_{v\in V'} \frac{Z_v}{4} = \frac{Z}{4}$ vertices must have been accessed in the external memory during the execution of $\A$.   
\end{proof}

\subsection{Proof of Lemma~\ref{lem:mainlb}}
\label{sec:mainlb}

Thus, it remains to prove \cref{lem:mainlb}.
The number of top nodes explodes very fast and
moreover, this number also depends on the order in which the top nodes are resolved. 
To bound it, we use a potential function argument, for which we  need to introduce two crucial notions.
A \emph{potential sequence (PS)} is a finite sequence of pairs of integers, e.g., $(t_1,\kappa_1),(t_2,\kappa_2), \dots$,
where $t_i \ge 0$  and 
$0 < \kappa_1 \le \kappa_2 \le \dots$. We call $\kappa_i$ the \emph{potential} and 
a PS is allowed to be empty, denoted by $()$.
A \emph{status vector} is a pair $(h; S)$, where $h > 0$ is an integer and $S$ is a PS.
The status vectors will help us bound the number of top nodes we will ever get in the adversary argument.
In particular, a status vector $W=(h;(t_1,\kappa_1), (t_2, \kappa_2),\dots)$ represents the end of an epoch
where we have at most $h$ top nodes and we have at most $t_i$ top nodes whose points have charge $\zeta-\kappa_i$. 
We make the following observations to simplify our mathematical manipulations of the status vectors.
\begin{observation}\label{ob:decomposition}
Let $S_1 = (t_1, \kappa_1), \dots (t_n, \kappa_n)$ and $S_2 = (t'_1, \kappa'_1), \dots, (t'_m, \kappa'_m)$ be two potential sequences, such that $\kappa_n \le \kappa'_1$. Then for any integer $\kappa$, $\kappa_n \le \kappa \le \kappa'_1$: 
\begin{itemize}
\item $(h; S_1, (0, \kappa), S_2) = (h; S_1, S_2)$, and 
\item $(h; S_1, (t, \kappa), S_2) = (h; (t', \kappa), (t-t', \kappa))$ for any integer $0 < t' < t)$.
\end{itemize}
\end{observation}
The first one states that since $(0, \kappa)$ represents 0 top nodes of charge $\zeta-\kappa$, they can be safely omitted from the status vector. The second one states that we can view a collection of $t$ top nodes as two collections of $t'$ and $t-t'$ nodes (with the same charges).

We claim that there is actually an explicit way to maximize the number of top nodes via
function  $\Phi$ defined as follows. Let $A(x) = A_{s+1}(x)$ and $S = (t_1, \kappa_1), (t_2, \kappa_2), \dots$ be an arbitrary PS. Then for all integers $t\ge 0$ and  $0 < \kappa \le \kappa_1$: 
\begin{align}
  \Phi(h; ()) & = h & \text{and} \label{eq:phi-base}\\
    \Phi(h;(t,\kappa), S) &= \left\{ 
      \begin{array}{ll}
        \Phi(h;S) & \text{if $t = 0$}, \\ 
        \Phi(A(h) ; (t-1,\kappa), S)  & \text{if $t > 0$ and $\kappa=1$}, \\ 
        \Phi(A(h) ; (A(h),\kappa-1), (t-1,\kappa), S)  & \text{if $t > 1$ and $\kappa>1$}. \\ 
      \end{array}\right. \label{eq:phi}
\end{align}
The $\Phi$ function captures (the upper bound on) the number of top nodes of various potential in the algorithm if the top node with the least potential are resolved first. In \cref{lem:max-nodes}, we will show that this resolution order produces the maximum number of top nodes.
So consider an epoch $i$ with at most $h_i$ top nodes. 
By \cref{ob:topnodes}, we can bound
$h_{i+1} < A_{s+1}(h_i) = A(h_i)$. 
\cref{eq:phi-base} captures the base case scenario where there are at most $h$ top nodes
and all the top nodes are terminated (represented by the empty status vector $()$).
\cref{eq:phi} captures three possible values of the first term $(t, \kappa)$ in the potential sequence of $\Phi$: 
\begin{itemize}
\item If $t = 0$, there are no nodes with potential $\kappa$, so this term can be ignored.
\item If there are $t$ nodes with potential $\kappa = 1$, i.e., the node contains items with charges $\zeta-1$, one of those nodes is terminated, resulting in the activation of nodes with $\zeta$ charge which, in turn, can only reach the termination condition.
Of course this still increases the number of top nodes to at most $A(h)$, of which $t-1$ top nodes have potential $\kappa = 1$.
\item In general, if all the top nodes have points with potential $\kappa > 1$, application of the resolution to one of the $t$ top nodes with the
smallest potential 
creates at most $A(h)$ nodes of potential $\kappa-1$ and leaves at most $t-1$ nodes of potential $\kappa$.
\end{itemize}

To prove that the $\Phi$ function is an upper bound on the number of top nodes, 
we need the following structural properties, which follow from the definitions of $\Phi$ and of the PS.
\begin{lemma}\label{lem:dec}
Let $S_1 = (t_1, \kappa_1), \dots, (t_n, \kappa_n)$ and $S_2 = (t'_1, \kappa'_1), \dots, (t'_m, \kappa'_m)$ be two potential sequences, such that $\kappa_n \le \kappa'_1$.
  Then, for any $h > 0$ and $\kappa_n \le  \kappa \le \kappa'_1$: $\Phi(h; S_1, S_2) = \Phi( \Phi(h;S_1); S_2)$ and $\Phi(h; S_1, (0, \kappa), S_2) = \Phi(h; S_1, S_2)$.
\end{lemma}

\begin{restatable}{lemma}{lemphi}\label{lem:zeta}
  For any integer $\kappa \ge 1$: $\Phi(h; (1,\kappa)) \le A_{s+2\kappa-1}(h)$.
\end{restatable}
\begin{proof} \emph{(By induction)}
  When $\kappa = 1$, for any $t > 0$: $\Phi(h; (t, 1)) = \Phi(A(h), (t-1, 1)) = 
  A^{(t)}_{s+1}(h)$ and so $\Phi(h; (1,1)) = A_{s+1}(h)$.
  Now assume that $\kappa \ge 2$ and that the claim is true for all positive integers $\kappa'< \kappa$.

  \begin{align*}
\Phi(h; (1,\kappa)) 
  & = \rlap{$\Phi( A(h); (A(h),\kappa-1), (0, \kappa)) = \Phi( A(h); (A(h),\kappa-1))$}
    & \text{by~\cref{eq:phi}}\\
  & = \Phi(A(h); \underbrace{(1, \kappa-1), \dots, (1, \kappa-1)}_{\text{$A(h)$ times}}) 
    & \text{by \cref{ob:decomposition}} \\
  & = \Phi(\Phi(A(h); (1,\kappa-1)); \underbrace{(1, \kappa-1), \dots, (1, \kappa-1)}_{\text{$A(h)-1$ times}}) 
    & \text{by \cref{lem:dec}} \\
  & \le \Phi(A_{s+2\kappa-3}(A(h));  \underbrace{(1, \kappa-1), \dots, (1, \kappa-1)}_{\text{$A(h)-1$ times}}) 
    & \text{by Inductive Hypothesis} \\
  & \le A^{(A(h))}_{s+2\kappa-3}(A(h)) 
    & \text{(*)} \\
  & < A_{s+2\kappa-2}(A_{s+1}(h)) < A_{s+2\kappa-2}(A_{s+2\kappa-2}(h))  
    & A(h) = A_{s+1}(h), \kappa \ge 2 \\
  &\le A_{s+2\kappa-1}(h),
    & A_{i}(A_i(x)) < A_{i+1}(x)
  \end{align*}
where (*) follows from a simple inductive argument.
\end{proof}

\begin{lemma}\label{lem:max-nodes}
    Let $V=(h ; (t_1,\kappa_1), (t_2,\kappa_2), \dots, (t_m,\kappa_m)$ 
    be a status vector that represents a state of the top nodes in the algorithm.
    Regardless of the order in which top nodes are resolved, the resulting
    number of top nodes is at most $\Phi(V)$.
\end{lemma}
\begin{proof}
Assume inductively that the claim is true for any $V'$ that is lexicographically greater than $V$. In the base case, the status vector is $(h'; ())$, i.e., the statement is vacuously true.
  Assume that the next node to get resolved has potential $\kappa_i$. 
  Let $V_i$ be the resulting status vector, i.e., 
  \[
      V_i = (A(h) ; (t_1,\kappa_1), \dots,  (t_{i-1},\kappa_{i-1}), (A(h),\kappa_{i}-1),(t_i-1,\kappa_i), (t_{i+1},\kappa_{i+1})\dots, (t_m,\kappa_m)).
\]
  Observe that no matter which top node the algorithm chooses to resolve first, the number of 
  top nodes at the next epoch will be upper bounded by $A(h) > h$, i.e, 
    $V_i > V$ lexicographically. Thus, for every $1\le i\le m$, inductively, $\Phi(V_i)$ is an upper bound for the number of top nodes.
    Observe that $\Phi(V) = \Phi(V_1)$, so to prove the claim it is sufficient to prove
    that $\Phi(V_i) \ge \Phi(V_{i+1})$ for every $1 \le i < m$. 

  Consider an arbitrary $1 \le i < m$. 
  Observe that if $\kappa_i = \kappa_{i+1}$ then $V_i = V_{i+1}$ and we have nothing to prove.
  Thus, assume $\kappa_{i+1} > \kappa_i$.
  Define the following three potential sequences:
  \begin{align*}
      \quad X =(t_1, \kappa_1), &\dots, (t_{i-1},\kappa_{i-1}), 
      \qquad\qquad\qquad\qquad Y =(A(h),\kappa_{i}-1),(t_i-1,\kappa_i), \\
      &\quad Z =(t_{i+1}-1, \kappa_{i+1}), (t_{i+2},\kappa_{i+2}), \dots, (t_{m},\kappa_{m}). & 
  \end{align*}
  Then, using \cref{ob:decomposition} we can rewrite $V_i$ and $V_{i+1}$ as follows:
\begin{align*}
    V_i     & = (A(h) ; X,Y,(1,\kappa_{i+1}),Z) \\
    V_{i+1} & = (A(h) ; X,(t_i,\kappa_{i}),(A(h),\kappa_{i+1}-1),Z) \\
            & = (A(h) ; X,(t_i-1,\kappa_{i}),(1,\kappa_{i}),(A(h),\kappa_{i+1}-1),Z).
\end{align*}
Observe that $\Phi$ is a strictly increasing function of all of its parameters. Then
\begin{align*}
    \Phi(V_{i+1})&< \Phi(A(h) ; X,(A(h),\kappa_{i}-1),(t_i-1,\kappa_{i}),(1,\kappa_{i}),(A(h),\kappa_{i+1}-1),Z) \\
        & = \Phi(A(h) ; X,Y,(1,\kappa_{i}),(A(h),\kappa_{i+1}-1),Z) \\
        & = \Phi( \Phi(A(h) ; X,Y); (1,\kappa_{i}),(A(h),\kappa_{i+1}-1),Z) \\
        & = \Phi(\chi; (1,\kappa_{i}),(A(h),\kappa_{i+1}-1),Z),
\end{align*}
where we set $\chi=\Phi(A(h) ; X,Y)$.  On the other hand, we have
\begin{align*}
    \Phi(V_i) &=   \Phi(A(h) ; X,Y,(1,\kappa_{i+1}),Z)  =  \Phi( \Phi(A(h) ; X,Y);(1,\kappa_{i+1}),Z)= \Phi(\chi; (1,\kappa_{i+1}),Z) \\
    &= \Phi(A(\chi); (A(\chi),\kappa_{i+1}-1),Z). 
\end{align*}
Observe that $\chi\ge A(h) > h$ which implies 
$A(\chi) > A(h)$ and thus $A(\chi) \ge A(h)+1$.
Thus,
\begin{align*}
    \Phi(V_i) &= \Phi(A(\chi); (A(\chi),\kappa_{i+1}-1),Z) \ge  \Phi(\chi; (A(\chi),\kappa_{i+1}-1),Z)  \\ 
     &\ge \Phi(\chi; (A(h)+1,\kappa_{i+1}-1),Z) =\Phi(\chi;(1,\kappa_{i+1}-1), (A(h),\kappa_{i+1}-1),Z) \\
     &\ge \Phi(\chi; (1,\kappa_{i}),(A(h),\kappa_{i+1}-1),Z) > \Phi(V_{i+1}), \tag{*} 
\end{align*}
where (*) follows from $\kappa_{i+1}-1 \ge \kappa_i$, and $\Phi$ being an increasing function.
\end{proof}

\section{Upper Bound for Planar Convex Hull}
\label{sec:cg}

In this section, we show that an algorithm similar to \Call{Maxima}{} can be used to compute the
output-sensitive convex hull of a planar point set.  
Recall that the algorithm starts with a (seed) parameter $h$. 
Observe that if $h\ge \sqrt{N}$, then we can simply switch to a worst-case $\OhOf{N\log N} = \OhOf{N\log h}$-time convex hull algorithm, e.g., Graham's scan, which can compute
the convex hull by sorting and scanning. 
Thus, in the rest of this section, we assume that $h < \sqrt{N}$.
The algorithm distributes the points into $2h$ buckets/subsets,
$P_1, \dots, P_{2h}$, of size roughly $\frac{N}{2h}$ and each pair of neighboring buckets $P_i$ and $P_{i+1}$ separated by a vertical line.
In the maxima algorithm, the next step (line~\ref{code:prune}) is to eliminate all points in each bucket that are dominated by the points in the buckets to its right.
In the context of the upper hull, this is slightly complicated and it involves pruning the points below the ``bridges'', which straddle the boundary between pairs of neighboring buckets $P_i$ and $P_{i+1}$.
We outline this process below.

To perform the pruning step, we adapt the strategy of
Kirkpatrick and Seidel's algorithm for the convex hull~\cite{kirkpatrick-seidel-CH}.
Their main technique is to partition the input points into two equal sets with a vertical line $\ell$ and then
find the ``bridge'', i.e., the edge of the upper hull intersected by $\ell$ in linear time. 
Here, we would like to generalize their technique to a setting that involves $2h$ subsets of points.

In our case, we have multiple buckets. 
Let $\ell_i$ be the vertical line separating (lies between) $P_i$ and $P_{i+1}$ and 
a bridge $b_i$ is the edge of the upper hull that intersects $\ell_i$ (it separates $P_i$ and $P_{i+1}$).
Let $Q$ be the set of all end points of all the bridges, $b_1, \dots, b_{2h-1}$.
Observe that pruning $P_i$ is trivial if all the \emph{bridges} 
intersecting $\ell_{i-1}$ and $\ell_i$ are already computed:
simply remove the points that lie below the bridges and
recurse on the remaining points in each subset $P_i$.
Thus, the pruning step can be reduced to a \emph{multi-bridge finding step} where the goal is to find all the bridges.

\begin{figure}[t]
  \centering
  \includegraphics[scale=0.5]{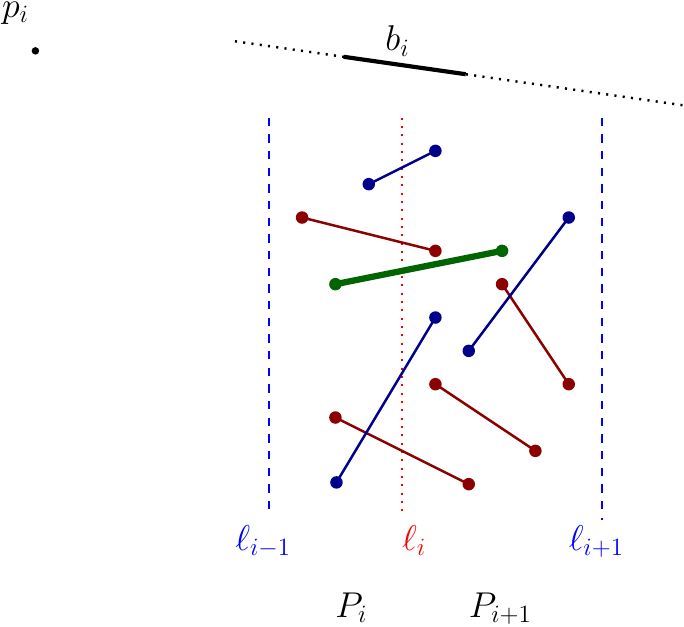}
  \caption{The collection $C_i$ initially contains all the points between $\ell_{i-1}$ and $\ell_{i+1}$. 
    The bridge $b_i$ intersects $\ell_i$ but could potentially extend beyond the collection in either directions
    which means the end points of $b_i$ are not guaranteed to be in the collection $C_i$.
    The green (thicker) pair is the pair with the median slope, $m_i$, among the pairs in the collection. 
    Red pairs have smaller slope than $m_i$, whereas blue pairs have larger slopes.}
  \label{fig:chprune}
\end{figure}

To find the bridges, we maintain $2h-1$ \emph{collections}:
The $i$-th collection, $C_i$, is initialized as $P_i \cup P_{i+1}$ and observe that each point belongs to at most
two collections. 
Then, we show that we can run the original Kirkpatrick and Seidel's algorithm on each collection.
However, there are some details to take care of.
During the algorithm we will be pruning a constant fraction of each $C_i$ until their total size becomes $\OhOf{N/\log N}$.
We will maintain the invariant that at all times 
\begin{align}
    Q \subset \bigcup_{i=1}^{2h-1} C_i. \label{eq:pruneinv}
\end{align}

The pruning strategy is as follows.
Let $C=\cup_{i=1}^{2h-1}C_i$.
First, we pair the points of $C_i$ arbitrarily and compute the median slope, $m_i$, of the pairs using the linear time and I/O median finding algorithm~\cite{selection}.
Next, we show that for each $m_i$, we can find an extreme point~$p_i$ of $C$ (not necessarily in $C_i$) in the direction orthogonal to $m_i$ in $\Split(N,h)$ time. 

\begin{lemma}\label{lem:ex}
  Given a set $P$ of $N$ points  and a set of $k$ slopes, for $k<\sqrt{N}$,
  we can compute the extreme points along each slope in $\OhOf{\Split(N,k)}$ time. 
\end{lemma}

\begin{proof}
  First, we sort (with an optimal CO algorithm) the slopes in decreasing order, $s_1, \dots, s_k$. Since $k$ is small this can be done
  in $\OhOf{\Scan(N)}$ time. 
  Next, we group the points arbitrarily into $\frac{N}{k^2}$ batches of size $k^2$ and compute the upper hull of each batch using Graham Scan. 
  This can be done in $\OhOf{\Split(N,k)}$ time by simple sorting and scanning each batch.

  Let $u_i$ be the upper hull of the $i$-th batch.
  Then, for each $u_i$, we do the following:
  we scan all $s_1, \dots, s_k$ at the same time as edges of $u_i$.
  Observe that the  upper hull edges of $u_i$ are also in the order of
  decreasing slopes. 
  This means that the vertices of $u_i$ that are extreme with respect to $s_1, \dots, s_k$ 
  can be found by just  a forward scan of $u_i$.
  While scanning the slopes we update each slope with the most extreme point seen so far.

  To analyze the I/O complexity, we need to consider two cases: 
  First, consider the case when  $k^2 \ge B$.
  In this case, scanning each $u_i$ forward requires $\OhOf{\Scan(k^2)}$ I/Os.
  As there are $k$ slopes, the total cost per batch is $\OhOf{\Scan(k^2)+1+\Scan(k)} = \OhOf{\Scan(k^2)}$.
  Over all $\frac{N}{k^2}$ batches this sums up to $\OhOf{\Scan(N)}$.
  Now assume $k^2<B$ and let $t = \left\lfloor \frac{B}{k^2}\right\rfloor$.
  In this case, $t$ consecutive batches are loaded with one I/O, and also all $k$ slopes are loaded with one I/O.
  Thus, we can find all  extreme points of $t$ consecutive batches at the same time, resulting in overall $\OhOf{\Scan(N)}$ I/O complexity.
\end{proof}

Thus, in the rest of the proof we assume that 
$p_i$, an extreme point of $C$ in the direction orthogonal to $m_i$ has already been computed. 

\begin{lemma}
    Given $p_i$ in each collection $C_i \subseteq C$, we can prune a fraction of the points in $C_i$ in $\OhOf{\scan{|C|}}$ I/Os.
    This process maintains \cref{eq:pruneinv}, meaning, the invariant still holds after pruning. 
\end{lemma}
\begin{proof}
    W.l.o.g., assume $p_i$ is to the left of $\ell_i$ (see \cref{fig:chprune}). 
    Observe that slopes of the bridges $b_1, \dots, b_{2h-1}$ is a decreasing sequence and
    since $C$ contains $Q$, it follows that the slope of  $b_i$ is at most $m_i$.
    We now apply the argument of Kirkpatrick and Seidel~\cite{kirkpatrick-seidel-CH}.
    Consider the pairs in $C_i$ that have slopes larger than $m_i$ (blue pairs in \cref{fig:chprune}).
    Each such pair is guaranteed to have a slope larger than or equal to the slope of $b_i$.
    We prune the left end points of each such pair as the left point cannot be an end point
    of $b_i$.

    However, we also need to verify that we have not pruned any point in $Q$.
    To do that, consider a point~$q \in Q$.
    By definition, there exists a collection $C_i$ such that $q$ is a vertex of the bridge $b_i$ over $\ell_i$.
    If $p_i$ is to the left (resp. right) of $\ell_i$, then it follows that $b_i$ does not have larger (resp. smaller) slope than $m_i$.
    Then, we consider the matched pairs in $C_i$ that have larger (resp. smaller) slope than $m_i$.
    For each matched pair of points $(p,p')$, where the point $p$ has smaller $X$-coordinate than $p'$, we pruned the point $p$ (resp. $p'$).
    However, observe that $q$ cannot be the pruned point. 
    This shows that our invariant is maintained. 
\end{proof}

After a pruning step, we go back to the beginning and pair the points in each $C_i$ arbitrarily again
and iterate. 
At each iteration, we will be pruning a  fraction of the points in each collection.
Thus, the size of the union, $C$,  will be decreasing geometrically.
Once the size decreases by a $\log N$ factor, we can simply compute the convex hull of the resulting
subset and find the bridges explicitly in $\OhOf{\sort{N/\log N}} = \OhOf{\scan{N}}$ I/Os.

\section{Lower Bound for Planar Convex Hull}
\label{sec:ch-LB}

In this section, we show how to adapt the adversarial argument from 
the maxima lower bound to the planar convex hull problem. 

Let  $\Delta$, $\sigma$ be some constants. We will use the notations $\dOmega{\cdot}$ and $\dOh{\cdot}$ to hide
factors that depend on $\Delta$ and $\sigma$; to avoid circular dependencies in the constants, we will use the notations to hide expressions
that depend only on $\Delta$ and $\sigma$ and fixed constants, i.e., constants that do not depend on other parameters in the construction. 
We call a polynomial $F$ \emph{suitable} if it has degree at most $\Delta-1$, and it is defined over $2\sigma$ indeterminates with real-valued coefficients.

\subparagraph*{Predicates.} 
At any moment, the algorithm can choose a suitable polynomial $F$ and a sequence $W$ of $\sigma$ input points, i.e., $W=p_1, \ldots, p_\sigma$,
and query the sign of the evaluation of $F$ on the $X$- and $Y$-coordinates of points $p_1, \ldots, p_\sigma$, i.e., query
$\sign(F(W))$.
We call $\sign(F(W))$ a \emph{geometric predicate}.

As with the previous lower bound, the adversary will decide the result of the predicate and then will inform the algorithm of the result.
Based on the result of the predicate, the algorithm can choose another predicate to be evaluated on a potentially different set of points. 
This process will continue until the algorithm can correctly declare the points on the convex hull.
Observe that a comparison is a simple predicate of degree 1 where $\Delta=1$ and  $\sigma=2$ and it involves only either $X$- or $Y$-coordinates of the two points.  
A standard orientation test is the sign of a 3$\times 3$ matrix that is obtained by placing the $X$- and $Y$-coordinates of three points
as the first two column and then adding a column of ones. Thus, in this case we have $\Delta=2$ and $\sigma=3$.
In this section, we show a lower bound that applies to any algorithm that uses such geometric predicates. 

\subparagraph*{A rough sketch of the proof.}
We show that the same set up that was used in \cref{sec:lowerbound} can be adapted to the convex hull problem as well.
In particular, we define a tree $T$ and the points will be placed on the nodes of $T$.
However, predicates make this process much more complicated because they reveal much more geometric information about the
position of the points than the relative order of the coordinates.
Nonetheless, for the argument in \cref{sec:lowerbound} to go through, we only need to replace two main tools used by the adversary:
one is the \emph{default strategy} used to resolve comparisons (predicates here) and two, \emph{node termination} 
which enables the adversary to essentially prune all but one point in a subtree of $T$.
To do these, we borrow an idea due to Afshani and Cheng~\cite{AC23,AC24} of embedding points very close
to the curve $Y=X^\Delta$. This enables us to approximate any monomial $X^iY^j$ as $X^{i+j\Delta}$ which in turn will convert any
bivariate polynomial of degree less than $\Delta$ into a univariate polynomial of degree less than $\Delta^2$ where distinct monomials
in the bivariate polynomial are mapped to distinct monomials in the univariate polynomials.
We follow this up by 
mapping every node~$v$ in $T$ to small enough geometric regions (squares), $Q(v)$, close to the curve $Y=X^\Delta$ such that 
from the point of view of the algorithm any point $p \in v$ can be anywhere inside~$Q(v)$, i.e.,  without violating any of the predicates chosen
by the algorithm.
Finally, we show that the resolution of predicates can be done by moving points only a constant number of levels down the tree $T$.
We now present the details.

\subparagraph*{The tree $T$.} Let $f=\Delta^2\cdot \Delta^{2\sigma} = \Delta^{2\sigma+2}= \dOh{1}$.
The adversary maintains a tree $T$ with fan out $f$, similar to the lower bound in \cref{sec:lowerbound}
but the adversary will also associate a square with every node of $T$.
Let $Q_0$ be a unit square with its center on the point $(1,1)$ which lies on the curve $Y=X^\Delta$.
$Q_0$ will be associated with the root of $T$.
We will shortly describe how the adversary assigns 
progressively smaller squares that are inside $Q_0$ to every node in~$T$.
We will use the notation $p \in v$ to denote that a point $p$ is placed in a node $v$ of $T$,
the square associated with $v$ is denoted by $Q(v)$ and the depth of $v$ (in $T$)  is denoted by $d(v)$.

\subparagraph*{Assigning squares.} Consider a non-root node $v$ at depth $d(v)$ in the tree $T$. The square~$Q(v)$ will have side length
$2^{-c^{d(v)}}$  where $c$ is a large enough value that will depend on $\Delta$ and~$\sigma$ (as mentioned, we will not ``hide'' $c$ in our asymptotic notations).
Consequently, the squares associated with the children of $v$ will have side length $2^{-c^{d(v)+1}}$ and they will be placed inside~$Q(v)$, spaced equally across the $X$-axis and 
centered on the curve $Y=X^\Delta$ (see \cref{fig:sq} for an example).
We claim that by picking $c$ large enough, we can guarantee that the squares assigned to the children of~$v$ will be fully inside $Q(v)$.
Let $A$ and $B$ be the intersection points of the boundary of $Q(v)$ with the curve~$Y=X^\Delta$.
Observe that the tangent to every point on the curve~$Y=X^\Delta$, which is inside $Q_0$, has a slope between $2^{-\Delta}$ and $1.5^\Delta$
which means that the difference between $X$- or the $Y$-coordinates of the points $A$ and $B$ is at least 
$\OmegaOf{2^{-\Delta} 2^{-c^{d(v)}}}$.
\begin{figure}[t]
    \centering
    \includegraphics[scale=0.7]{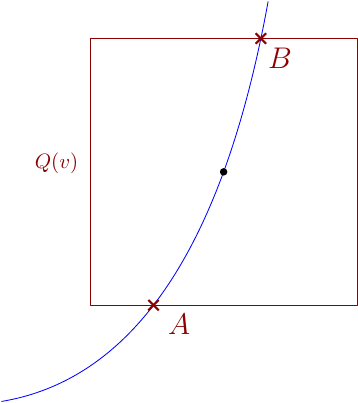}
    \caption{Square $Q(v)$ for node $v$.}
    \label{fig:sq}
\end{figure}
By setting~$c$ large enough we can easily ensure that $f \cdot 2^{-c^{d(v)}} \ge 2^\Delta \cdot f \cdot 2^{-c^{{d(v)}+1}}$. This means that there is enough space inside $Q(v)$ to place the squares of the children of $v$.

As before, any non-empty node $v$ whose ancestors do not contain any points is called a \emph{top node}.
The equivalent of \cref{inv:tree} in this case is the following invariant.
\begin{inv}[The tree invariant]\label{inv:treeinv2}
    For every point $p \in v$, every placement of the point $p$ inside the square $Q(v)$ is consistent with the results of all 
    previous predicates returned by the adversary, meaning, if the algorithm in the past has queried the sign of
    a suitable polynomial $F$ on $\sigma$ points $p_1 \in v_1,, \dots, p_\sigma \in v_\sigma$  and the adversary has returned a sign value $z$, then
    for every point $W \in Q(v_1) \times \dots \times Q(v_\sigma)$,  we have \textup{$\sign(F(W)) = z$}.
\end{inv}

Node termination is captured with the following definition. 

\begin{definition}[Node termination]
  The adversary can \emph{terminate} a top node $v$ by
  picking four points $p_1,p_2,p_3,p_4 \in u$, and declaring that they lie on the corners of 
  $Q(v)$. The adversary also picks some arbitrary coordinates for the points in $T(v)$ within
  their corresponding squares and declares these to the algorithm.
\end{definition}

Note that similar to the previous proof, after the termination of a node $v$, the positions
of all the points in the subtree of $v$ will be known to the algorithm and thus 
we can assume that no predicates involves the coordinates of such points. 
We call these the \emph{pruned points} and any other remaining point an \emph{unpruned points}.

Thus, the heart of the construction is that the adversary can replace the \emph{default strategy} with a similar strategy, captured in the following lemma. 

\begin{lemma}\label{lem:chresolve}
    Consider a predicate $F$ chosen by the algorithm on a set of $\sigma$ points $p_1, \dots, p_\sigma$ such that
    $p_i \in v_i$.
    The adversary can maintain \cref{inv:treeinv2} by moving each  point $p_i$ to a node $u_i$ such that
    $u_i$ is a $t_i$-descendant of $v_1$ for some  $t_i \le \sigma+1$.

    In fact, we show something stronger. The adversary can create an infinite sequence of positive real values 
    $\gamma_1, \gamma_2, \dots$ such that the following holds:
    For every point $W \in Q(u_1)\times \dots \times Q(u_\sigma)$, $|F(W)| \ge \gamma_d$ where
    $d= \max\left\{ d(u_1), \dots, d(u_\sigma) \right\}$. 
\end{lemma}

\ifArxiv
We prove the above lemma in \cref{sec:chresolve}.
\else
\begin{proof}
The proof is an adaptation of an idea due to Afshani and Cheng~\cite{AC23,AC24} of embedding the points very close
to the curve $Y=X^\Delta$. 
The details are presented in~\cite{full-version}.
\end{proof}
\fi

Observe that as the function $F$ is continuous, the latter claim in the above lemma not only implies that $F$ is non-zero
over $Q(u_1)\times \dots \times Q(u_\sigma)$ but also that its magnitude is lower bounded by some fixed parameter that only
depends on the depth of the deepest node among $u_1, \cdots, u_\sigma$.

We now show that these modifications are sufficient to obtain a lower bound for the convex hull problem.

\begin{theorem}\label{thm:chlb}
    Consider a planar point set $P$ of size $N$ with $H$ convex hull points and let  $\Delta > 0$ and $\sigma > 0$ be integer constants.  
    Consider an algorithm that computes the convex hull of $P$ using predicates which are polynomials of degree less than $\Delta$
    with $2\sigma$ indeterminates, where each predicate is applied to $\sigma$ input points.
    If the algorithm uses $\OhOf{N A_s(H)}$ predicates, then it requires
  $\OmegaOf{\frac{N}{B}(\alpha_1(\min\left\{ H, \sqrt{\frac{N}{M}}\right\})-s)}$ I/Os. 
\end{theorem}
\begin{proof}

    The proof is very similar to the proof of \cref{thm:maximalb}. 
    The observation is that in the proof of \cref{thm:maximalb}, only
    two operations, the default strategy and node terminations, give information to the algorithm about the positions of the points and the rest of the proof
    does not really use the fact that we are computing the maxima other than that the adversary maintains a binary tree and each comparison
    moves the points involved in the comparison at most 1 level down the tree. 
    In our case, we are maintaining a tree $T$ of fanout $f = \Delta^{2\sigma+2}$ but it can be simulated with a binary tree $T'$ where 
    each level of $T$ corresponds to $\log(f) = (2\sigma+2)\log(\Delta) = \Theta(\sigma\log(\Delta))$ levels of $T'$.
    \cref{lem:chresolve} directly replaces the \emph{the default} strategy and it can be used by the adversary to resolve
    the predicates.
    However, the resolution of the predicates moves $\sigma$ points, $\sigma+1$ levels down the tree which in the corresponding binary tree $T'$
    corresponds to moving $\sigma$ points $\OhOf{\sigma\log(\Delta)}$ levels down. 

    Next, consider node termination.
    Terminating a node $u$ ensures that the points in $T(u)$ will contribute at most four vertices to the convex  hull. 
    Consequently, if four times the number of top nodes will be an upper bound on the size of the convex hull which again changes only
    the argument in \cref{thm:maximalb} by a constant factor. 

    The rest of the proof is, therefore, identical to the proof of \cref{thm:maximalb}, except any one predicate translates to $\OhOf{\sigma^2\log\Delta} = \dOh{1}$ comparisons, 
    which only changes the internal memory work of the algorithm by a constant factor, and  the same lower bound argument applies. 
\end{proof}

\ifArxiv
\subsection{Proof of Lemma~\ref{lem:chresolve}}
\label{sec:chresolve}

The proof is quite technical and we adapt an idea due to Afshani and Cheng~\cite{AC23,AC24} of embedding the points very close
to the curve $Y=X^\Delta$. 
The details are as follows.
We will use the following theorem proven in~\cite{AC24}.

\begin{restatable}{lemma}{genbase}\label{lem:gen1d}
Consider two univariate polynomials on $x$, 
$P(x)=\sum_{i=0}^\beta a_ix^i$, and 
$Q(x)=\sum_{i=0}^\beta b_ix^i$ of (constant) degree $\beta$.
Assume 
    there is some $i$, $0\le i \le \beta$ such that $|a_i-b_i|\ge\eta$ and
  an interval $\I \subset \R$  such that for every
    $x_0 \in I$ we have $|P(x_0)-Q(x_0)| \le w$.
    Then, length of  $\I$ 
is upper bounded by $\OhOf{(w/\eta)^{1/\U}}$, 
where $\U=\binom{\beta+1}{2}$ and the $\OhOf{\cdot}$-notation hide constant
factors that depend on $\beta$.

\end{restatable}

We note that the proof of the above lemma does not in fact use the fact that the two polynomials $P$ and $Q$
have the property that $|P(x_0)-Q(x_0)| \le w$ for all the points $x$ in the interval $\I$. 
In particular, the only requirement is that one can find $\beta+1$ values $x_0$ that are sufficiently far apart 
such that $|P(x_0)-Q(x_0)| \le w$.
Due to this, we will in fact use the following formulation of the above lemma. 

\begin{restatable}{lemma}{genbase2}\label{lem:uni}
Consider two univariate polynomials on $x$, 
$P(x)=\sum_{i=0}^\beta a_ix^i$, and 
    $Q(x)=\sum_{i=0}^\beta a_ix^i$ of (constant) degree $\beta$
    such that for  some $i$, $0\le i \le \beta$ we have $|a_i-b_i|\ge\eta$.
    Consider a sequence of values $x_0, \dots, x_\Delta$ such that
    for every $i\ge 1$, we have $x_i \ge x_{i-1} + \mu$ for some parameter $\mu$.
    Then, there exists a value $i$ such that 
    $|P(x_i)-Q(x_i)| = \OmegaOf{\eta \mu^{\U}}$,
    where $\U=\binom{\beta+1}{2}$ and the $\OmegaOf{\cdot}$-notation hides constant
    factors that depend on $\beta$.
\end{restatable}

    W.l.o.g., we can assume that the sequence of nodes $v_1, \dots, v_\sigma$ have non-increasing depth values (otherwise, we can just relabel them
    decreasingly based on depth).
    Let $v'_i$ be some $(\sigma-i)$-descendant of $v_i$.
    Observe that the sequence $v'_1, \dots, v'_\sigma$ has strictly decreasing depth.
    See \cref{fig:vprime}.
    For each $i, 1 \le i \le \sigma$, the node $u_i$ claimed in the lemma will be a carefully chosen child of $v'_i$.
\begin{figure}[ht]
    \centering
    \includegraphics[scale=0.75]{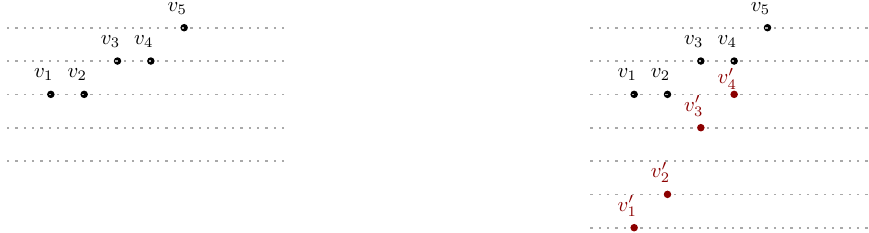}
    \caption{(left) The sequence $v_1, \dots, v_5$ has non-increasing depth. The depth of a node is visualized with a
    gray dashed line. (right) $v'_1$ is some 4-descendant of $v_1$, $v'_2$ is some 3-descendant of $v_2$ and so on. The
resulting sequence $v'_1, \dots, v'_5$ has strictly decreasing depth.}
    \label{fig:vprime}
\end{figure}

    In the remainder of the proof, we will explain how this choice is made.

    Decompose the polynomial $F$ recursively as follows.
    Define $G_\sigma = \left\{ F \right\}$, i.e., a set of polynomials that only contains $F$.
    We define $G_{i-1}$ using $G_i$, assuming $G_i$  is a set of polynomials in $X_1, Y_1, X_2, Y_2, \dots, X_i, Y_i$.
    For every polynomial $P \in G_i$, write $P$ as a polynomial in $X_i$ and $Y_i$ with coefficients
    which are polynomials in $X_1, Y_1, \dots, X_{i-1},Y_{i-1}$; add the coefficients to $G_{i-1}$
    and repeat this for every $P \in G_i$.
    Observe that $G_1$ a set of bivariate polynomials in $X_1$ and $Y_1$ and
    $G_0$  is defined as the set of all the coefficients of the polynomials in $G_1$. 
    Finally, note that as $G_\sigma$ is initialized to contain only $F$, $G_0$ contains the coefficients of $F$ and 
    since $F$ has degree less than $\Delta$ on $2\sigma$ indeterminates, $F$ has fewer than
    $\Delta^{2\sigma}$ monomials and thus for every $j, 0 \le j \le \sigma$, we have
    $|G_j| < \Delta^{2\sigma}$.

    Our proof strategy is to show that we can pick the nodes $u_1, \dots, u_\sigma$, step-by-step 
    to satisfy the requirements in the lemma. 
    We explain the details below.

    \subparagraph*{Step $i$ hypothesis.}
    Let $\Q = Q(u_1)\times \dots \times Q(u_{i})$.
    We assume that we have selected the nodes $u_1, u_{2}, \dots, u_{i}$ such that 
    the following holds. 
    For every polynomial $H$ in $G_i$ (which is a polynomial on $W=(X_1,Y_1, \dots, X_{i},Y_{i})$)
    there exists (fixed) real values $M_H$ and $\varphi_H$ such that the function $\lambda_H(W)$ 
    defined as $\lambda_H(W)  =H(W) - M_H $ has the property that
    for all $W \in \Q$ we have
    $|\lambda_H(W)| \le \varphi_{H}$.
    Define $R_H = \frac{M_H}{\varphi_H}$ 
    and $\bR_i = \min_{H \in G_i}R_H$.

    \subparagraph*{The base case (step $0$ hypothesis).}
    Step $0$ is our base case since we are considering the set $G_0$ which is simply a set of
    of constants and thus each $H \in G_0$ is already a fixed value. 
    Thus, we can take $M_H = H$ and $\varphi_H = 0$.
    Observe that $\bR_{0}$ is unbounded. 

\begin{figure}[h]
    \centering
    \includegraphics[scale=0.7]{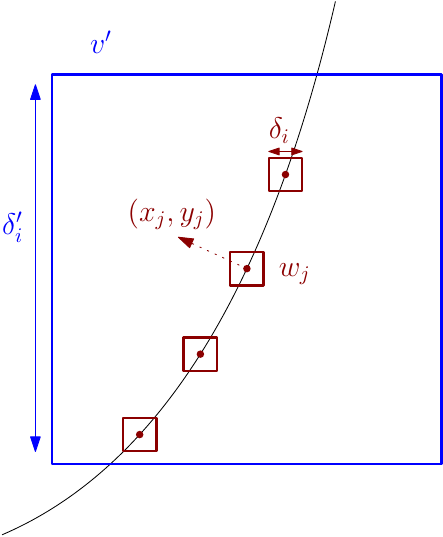}
    \caption{}
    \label{fig:chcurve}
\end{figure}

    \subparagraph*{The inductive step.}
    In this case, we assume step $i-1$ hypothesis holds and we show that we can select
    $u_i$ such that step $i$ hypothesis also holds and based on the value of $\bR_{i-1}$ we obtain a
    lower bound for the value of $\bR_i$.
    To do that, we need to choose $u_i$ as a child of $v'_i$ in a careful way. 
    To reduce the amount of indices, let $v'=v'_i$ and let $w_1, \dots, w_f$ be the children of
    $v'$.
    Consider a polynomial $H \in G_{i}$. 
    Consider $H$ as a polynomial in $(X_i,Y_i)$ with coefficients which are
    polynomials in $G_{i-1}$. 
    Let $\delta'_i$ be the side length of the square $Q(v'_i)$ and let $\delta_i$ be the side length of squares 
    $Q(w_j)$, $1 \le j \le f$. 
    Consider a child $w_j$ and let $(x_j,y_j)$ be the center of $Q(w_j)$ which lies on the curve $Y=X^\Delta$.
    Observe that for every point $(X,Y) \in Q(w_j)$ we have $|X-x_j|,|Y-y_j| \le \delta_i$
    and also $y_j = x_j^\Delta$.
    This in turn means that for every monomial $X^\ell Y^k$, we have 
    \begin{align}
        X^\ell Y^k = x_j^{k\Delta+\ell} + \varepsilon_{\ell,k}(X,Y) \quad \quad \mbox{for all } (X,Y) \in Q(w_j), \label{eq:approx}
    \end{align}
    where $\varepsilon_{\ell,k}(X,Y)$ is a function and observe that $|\varepsilon_{\ell,k}(X,Y)| = \dOh{\delta_i}$
    for every point $(X,Y) \in Q(w_j)$.
    We plug in this approximation in the polynomial $H$.
    First, we write $H$ as a polynomial in $X_i,Y_i$ as 
    \begin{align*}
      H(X_i,Y_i) =      \sum_{\ell,k, \ell+k<\Delta}\alpha_{\ell,k} X_i^\ell Y_i^k,
    \end{align*}
    where each $\alpha_{\ell,k}$ is (a polynomial) in $G_{i-1}$.
    We plug in the approximation from \cref{eq:approx} to obtain
    \begin{align*}
        H(X_i,Y_i) =      \sum_{\ell,k, \ell+k<\Delta}\alpha_{\ell,k}  (x_j^{k\Delta+\ell} + \varepsilon_{\ell,k}(X_i,Y_i)) \quad\quad  \mbox{for all }(X_i,Y_i) \in Q(w_j).
    \end{align*}
    We now use the step $i-1$ hypothesis. 
    To avoid adding extra notation, define $M_{\ell,k} = M_{\alpha_{\ell,k}}$ 
    $\lambda_{\ell,k} = \lambda_{\alpha_{\ell,k}}$  and
    $\varphi_{\ell,k} = \varphi_{\alpha_{\ell,k}}$ 
    which we know to exist from the step $i-1$.
    Note that both $\alpha_{\ell,k}$ and $\lambda_{\ell,k}$ are functions of $X_1, Y_1, \dots, X_{i-1}, Y_{i-1}$ but
    this dependency is not reflected in our notation for brevity. 
    Thus, we can write  $\alpha_{\ell,k} = M_{\ell,k} + \lambda_{\ell,k}$  
    and by the step $i+1$ hypothesis  $|\lambda_{\ell,k}| \le \varphi_{i+1}$.
    Thus, for every $(X_i,Y_i) \in Q(w_j)$ we get that
    \begin{align}
        H(X_i,Y_i) &=      \sum_{\ell,k, \ell+k<\Delta}(M_{\ell,k}+\lambda_{\ell,k})  (x_j^{k\Delta+\ell} + \varepsilon_{\ell,k}) \nonumber \\
        &=     \left(\sum_{\ell,k, \ell+k<\Delta}M_{\ell,k}x_j^{k\Delta+\ell} \right)  \label{eq:poly} \\
        &+ \sum_{\ell,k, \ell+k<\Delta}\left(\lambda_{\ell,k} (x_j^{k\Delta+\ell} + \varepsilon_{\ell,k}) + M_{\ell,k}\varepsilon_{\ell,k}\right).\label{eq:err}
    \end{align}
    Observe that  if we fix $j$, then \cref{eq:poly} is a fixed value (i.e., does not depend on 
    $X$ or $Y$) and thus we can set $M_{H}$ to be the value given by \cref{eq:poly}.
    We consider \cref{eq:err} our error term which in turn defines the function $\lambda_H$.
    The maximum magnitude of $\lambda_H$ will be the value of $\varphi_H$.
    Let $M = \max_{\ell,k}\left\{ M_{\ell,k} \right\}$.
    Observe that in \cref{eq:err}, the value $x_j$ is upper bounded by $2$ and the magnitude of the $\varepsilon_{\ell,k}$ 
    is upper bounded by $\dOh{\delta_i}$.
    In addition, via the definition of $\bR_{i-1}$ we have
    \begin{align*}
        \bR_{i-1} \le \frac{|M_{\ell,k}|}{|\lambda_{\ell,k}|} \le \frac{M}{|\lambda_{\ell,k}|}, 
    \end{align*}
    which implies 
    \begin{align}
        |\lambda_{\ell,k}| \le \frac{M}{\bR_{i-1}}. \label{eq:lammag}
    \end{align}
    Since this holds for every $\lambda_{\ell,k}$ we have
    \begin{align}
        \varphi_H \le \frac{M}{\bR_{i-1}} \label{eq:errphi}.
    \end{align}
    We use these approximations in \cref{eq:err}.
    First observe that 
    \begin{align}
        \Bigl|\sum_{\ell,k, \ell+k<\Delta}  M_{\ell,k}\cdot \varepsilon_{\ell,k}\Bigr| 
        \le  \sum_{\ell,k, \ell+k<\Delta} M\cdot |\varepsilon_{\ell,k}| 
        \le \dOh{M\delta_i}.\label{eq:part1}
    \end{align}
    Next, observe that $|x_j^{k\Delta+\ell}+\epsilon_{\ell,k}| = \dOh{1}$ since $\delta_i < 1$.
    Thus, using \cref{eq:lammag} we get that
    \begin{align}
        \Bigl|\sum_{\ell,k, \ell+k<\Delta}\lambda_{\ell,k} (x_j^{k\Delta+\ell} + \varepsilon_{\ell,k}) \Bigr|
        \le \sum_{\ell,k, \ell+k<\Delta}\dOh{|\lambda_{\ell,k}|}
        = \dOh{\frac{M}{\bR_{i-1}}}.\label{eq:part2}
    \end{align}
    Recall that \cref{eq:err} defines the function $\lambda_{H}$ and thus
    the $|\lambda_H|$ can be upper bounded by summing up \cref{eq:part1} and \cref{eq:part2} which in turn implies that
    for every $(X_i,Y_i) \in Q(w_j)$ we have
    \begin{align}
        |\lambda_H| \le \dOh{\frac{M}{\bR_{i-1}}+ M \delta_i} \label{eq:errmag}.
    \end{align}

    The above approximations are valid with respect to every choice of $j$ but we have not discussed how do we make the
    choice of $j$. 
    The idea is that we can choose an index $j$ such that 
    the value $|M_H|$ is much larger than the value given by \cref{eq:errmag} and then we argue that such a choice will give us
    what we want.
    To choose $j$, 
    we apply the \cref{lem:uni} to \cref{eq:poly}. Observe that
    $v$ has $f$ children whose squares are centered on $x_1, \dots, x_f$ and that the distance between $x_j$ and $x_{j+1}$ is at least $\delta'_i/(2f)$.
    Also note that  if $(\ell,k) \not = (\ell',k')$ then we have $\ell+k\Delta \not= \ell' + k'\Delta$ since
    both $\ell$ and $k$ are strictly smaller than $\Delta$.
    Thus, we can consider \cref{eq:poly} as the evaluation of a univariate polynomial of degree less than $\Delta^2$ with at least one 
    coefficient that has magnitude $M$.
    By \cref{lem:uni},  there can be fewer than $\Delta^2$ points $x_j$ such that the evaluation of the polynomial has 
    magnitude at most $\dOh{M(\delta'_i/f)^{\Delta^4}} = \dOh{M \delta'^{\Delta^4}_i}$. 
    We call any such child $w_j$ \emph{bad (wrt $H$)}, otherwise it is \emph{good}.
    Observe that if $w_j$ is good, then we have
    \begin{align}
        \frac{M_H}{\varphi_H} 
        \ge \dOmega{\frac{ M \delta'^{\Delta^4}_i}{\frac{M}{\bR_{i-1}}+M\delta_i}}  
        \ge \dOmega{\min\left\{ \bR_{i-1}\delta'^{\Delta^4}_i,\frac{\delta'^{\Delta^4}_i}{\delta_i} \right\}} \label{eq:rrec}.
    \end{align}
        
    Observe that for every $H \in G_{i}$, we get fewer than $\Delta^2$ bad children. 
    Since $|G_i| < \Delta^{2s}$ and $f \ge \Delta^2 \Delta^{2s}$, there exists a child $w_j$ which is good with respect to all the polynomials $H \in G_i$. 
    This is the child that we choose and thus we set $u_i = w_j$.
    This implies that the \cref{eq:rrec} holds with respect to every $H\in G_i$ and thus by definition of $\bR_i$ we get that
    \begin{align}
        \bR_{i} \ge    \min\left\{\mu \bR_{i-1}\delta'^{\Delta^4}_i,\mu\frac{\delta'^{\Delta^4}_i}{\delta_i} \right\} \label{eq:polyrec}.
    \end{align}
    Observe that this defines a recursion where
    $\bR_0$ is unbounded (meaning, can be set to any arbitrary large value) and $\mu$ is the hidden constant in the $\dOmega{\cdot}$ notation
    (which depends on $\sigma$ and $\Delta$).
    We would like to obtain a lower bound for $\bR_\sigma$.
    Assume $j$ is the largest index in the sequence $\bR_\sigma, \bR_{\sigma-1}, \dots, \bR_1$ where 
    $\bR_{j-1}\delta'^{\Delta^4}_j > \frac{\delta'^{\Delta^4}_j}{\delta_j}$.
    Since $\bR_0$ is unbounded, such an index always exists.
    This implies that $\bR_{j} \ge \frac{\delta'^{\Delta^4}_j}{\delta_j}$ and in the remaining terms $\bR_{j+1}, \dots, \bR_\sigma$, the first term in the minimization 
    in \cref{eq:polyrec} achieves the minimum.
    Thus, we get that 
    \begin{align}
        \bR_{\sigma} \ge  \min_{j=1, \cdots, \sigma}\left\{    \mu^{\sigma-j} \frac{\left( \delta'_j \delta'_{j+1} \ldots \delta'_\sigma \right)^{\Delta^4}}{\delta_j} \right\}.
    \end{align}
    We now use the property that the nodes $u_j, \dots, u_\sigma$ have strictly decreasing depth. 
    Note that this also implies that sequence $\delta'_j, \dots,\delta'_\sigma$ is strictly increasing
    and thus $\delta'_j$ is the smallest among them. 
    Because of this, we can lower bound each of the terms $\delta'_{j+1}, \cdots, \delta'_{\sigma}$ with $\delta'_j$.
    Thus,
    \begin{align}
        \bR_{\sigma} \ge  \min_{j=1, \cdots, \sigma}\left\{    \mu^{\sigma-j} \frac{\left( \delta'^{\sigma-j+1}_j  \right)^{\Delta^4}}{\delta_j} \right\}.
    \end{align}
    Let $d_j$ be the depth of node $u_j$.
    We thus get that $\delta'_j = 2^{-c^{d_j}}$ but also that
    $\delta_j =  2^{-c^{d_j+1}}$ since $\delta_j$ is the side length of the squares assigned to the children of
    $u_j$.
    We thus get that
    \begin{align*}
           \mu^{\sigma-j} \frac{\left(\delta'^{\sigma-j+1}_j\right)^{\Delta^4}}{\delta_j} \ge \mu^{\sigma-j} \frac{\left( 2^{- c^{d_j}(\sigma-j+1)}\right)^{\Delta^4}}{2^{-c^{d_j+1}}} &=
        \mu^{\sigma-j} \frac{ 2^{- c^{d_j}(\sigma-j+1)\Delta^4}}{2^{-c^{d_j+1}}} \\
            &=\mu^{\sigma-j} \frac{ 2^{- c^{d_j}(\sigma-j+1)\Delta^4}}{2^{-c^{d_j} \cdot c}}.
    \end{align*}
    Now observe that by setting $c$ a large enough constant (depending on $\Delta$ and $\sigma$), we can guarantee that $\bR_\sigma \ge 10$ for every possible choice of $j$.
    Note that this implies that for the function $F$, there exists fixed values $M_F$ and $\varphi_F$ 
    such that for every $W \in Q(u_1)\times \dots \times Q(u_\sigma)$, for the function $\lambda_F = F(W)- M_F$ we have $\frac{M_F}{\lambda_F} \ge 10$.
    This implies that $F$ has the same sign for every such point $W$ and thus it is resolved.

    Finally, it remains to show the stronger property claimed in the lemma.
    Let $M_i = \min_{H \in G_i} |M_H|$.
    Observe that $M_1 \ge 1$ since $F$ is suitable.
    Since we have selected a good node $w_j$, for every $H \in G_i$ we have
    $M_H \ge \dOmega{M_{i-1} \delta'^{\Delta^4}_i}$ and thus
    $M_i \ge \dOmega{M_{i-1} \delta'^{\Delta^4}_i}$.
    This in turns gives another recursion for $M_i$ and thus $M_\sigma$ will solved to a value that depends on~$\sigma$.
    This in turn defines the sequence $\gamma_i$ claimed in the lemma.

\fi

\section{Randomized Algorithm}
\label{sec:rand}

In this section we present a randomized algorithm (\cref{alg:random}) for computing the maxima.
The algorithm is extremely similar to the deterministic algorithm of \cref{alg:upper} with only two differences.
In the original algorithm, first $2h$ subproblems are created and then they are pruned. 
Then, the algorithm does a recursion on the subproblems, in the order $P_1, \ldots, P_{2h}$.
The first difference is that in the modified algorithm, 
with 50\% probability, instead of that we will recurse in the order $P_{2h}, \ldots, P_1$. 
The second difference is that the algorithm is called with the seed $s=2$ at the top level. 
Below, we show that these simple changes lead to an algorithm that has optimal (worst-case) time complexity and optimal expected I/O
complexity, i.e., an optimal randomized cache-oblivious algorithm. 

\begin{algorithm*}[t]
\begin{algorithmic}[1]
\Procedure{Random-Maxima}{$P, h$} 
 \If {$|P| \le 1$}  \Comment{Base case} 
  \State Output $P$ and \Return $|P|$
\EndIf
\State $\Ell =(P_{2h}, \dots, P_{1}) = \Call{Dist}{P,2h}$\Comment{{Distribute $P$ into $2h$ buckets w.r.t.\ $X$-coor.}}\label{code:random-distribute} 
\State $\Call{Prune}{\Ell}$ \Comment{For each $P_i$ prune points dominated by points in $\bigcup_{j=1}^{i-1} P_j$}\label{code:random-prune}
\State Reverse $\Ell$ with probability $\frac 12$ \Comment{Backwards for-loop with prob. $\frac 12$}
\State $\Hbar=0$
\For {each $P_j \in \Ell$}
    \State $H_i = \Call{Random-Maxima}{P_j,h+\Hbar}$
    \State $\Hbar = \Hbar + H_i$
\EndFor
\State \Return $\Hbar$
\EndProcedure
\end{algorithmic}
\caption{Randomized algorithm for computing the maxima of a planar point set $P$ with an integer seed $h\ge 2$.} 
\label{alg:random}
\end{algorithm*}

First observe that the time analysis of \cref{alg:upper} is still valid for the randomized algorithm because 
the order of recursive calls on the subproblems is not relevant to the internal computation time. 
This implies that the randomized algorithm has an optimal worst-case running time. 

As before, let $H$ be the total number of maxima points in the entire input set $P$, 
$N = |P|$, $s$ be the initial seed given to the 
algorithm (i.e., at the top level, the algorithm is called with $h=s = 2$), 
and $H_i$ denote the number of maxima points of $P_i$ (computed during the recursive call on $P_i$).
Define $H^\oplus = H + s$, and
$G_i = h + H_1 + \cdots + H_{i-1}$ and $G'_i = h + H_{i+1} + \cdots + H_{2h}$.
Observe that a forward \textbf{for} loop calls $\Call{Random-Maxima}{P_i,G_i}$, while the
backward \textbf{for} loop calls $\Call{Random-Maxima}{P_i,G'_i}$. Each case happens with 50\% probability, therefore, 
the expected I/O complexity of \Call{Random-Maxima}{$P,h$} is defined by the following recurrence relation
(with the base cases being the same as in \cref{eq:main-detailed}): 
\begin{align}
  \Trand_H(N, h) = \left\{
    \begin{array}{ll}
      \OhOf{\Split(N,h)}  & \text{if } h \ge H, \\
      \OhOf{\Split(N,H)}  & \text{if } H > h \ge m, \\
      \frac{1}{2} \cdot \sum_{i = 1}^{2h} \Trand_{H_i}\left(\frac{N}{2h}, G_i\right) 
         +\frac{1}{2} \cdot \sum_{i = 1}^{2h} \Trand_{H_i}\left(\frac{N}{2h}, G'_i\right)
        \rlap{$\mathrel{+}\OhOf {\Split(N,h)}$} \\
        & \text{if } h < \min(H,m).
    \end{array}
  \right.\label{eq:mainr} 
\end{align}

Let $\imax = \argmax\{H_1, \dots, H_{2h}\}$, i.e., the index of the largest  $H_1, \ldots, H_{2h}$.
The main observation here is that for any index $j \not = \imax$, either $G_j \ge H_j$, or
$G'_j \ge H_j$ (both can be true, and this might also hold for $H_\imax$), i.e., 
the I/O complexity of at least one of the recursive calls 
$\Call{Random-Maxima}{P_j,G_j}$ or 
    $\Call{Random-Maxima}{P_j,G'_j}$ is covered by the base cases. 
For instance, if $G_j \ge H_j$, 
$\Trand_{H_j}(N/(2h),G_j) = \OhOf{\Split(N/(2h),G_j)} = \OhOf{\Split(N/(2h), \Ho)}$ (the last equality follows from $G_j \le \Ho$).
Similarly, when $G'_j \ge H_j$, 
$\Trand_{H_j}(N/(2h),G'_j) = \OhOf{\Split(N/(2h),\Ho)}$.
Then we can rewrite $\Trand_H(N,h)$ as follows by combining the recursive calls defined by the base cases.
\begin{align*}
  \Trand_H(N, h) 
    & = \frac{1}{2} \cdot \left(\sum_{i = 1}^{\imax-1} \Trand_{H_i}\left(\frac{N}{2h}, G_i\right) + \Trand_{H_\imax}\left(\frac{N}{2h}, G_\imax\right)+ \sum_{i = \imax+1}^{2h} \Trand_{H_i}\left(\frac{N}{2h}, G_i\right)\right)\\
    & \mbox{}\hspace{-1.cm} + \frac{1}{2} \cdot \left(\sum_{i = 1}^{\imax-1} \Trand_{H_i}\left(\frac{N}{2h}, G'_i\right) + \Trand_{H_\imax}\left(\frac{N}{2h}, G'_\imax\right)+ \sum_{i = \imax+1}^{2h} \Trand_{H_i}(\frac{N}{2h}, G'_i) \right)
      + \OhOf{\Split(N,h)} \\
    & \mbox{}\hspace{-1.1cm} \le \frac{1}{2} \cdot\left(\Trand_{H_\imax}\left(\frac{N}{2h},G_\imax\right) +\Trand_{H_\imax}\left(\frac{N}{2h},G'_\imax\right) \right) +    \frac{1}{2} \cdot\sum_{i = 1}^{2h} \Trand_{H_i}\left(\frac{N}{2h}, G^{\star}_i\right)   + c \cdot {\Split(N,\Ho)},
\end{align*}
where $G^{\star}_i$ is either $G_i$ or $G'_i$, $c > 0$ is some constant,  and the last inequality follows from $h \le \Ho$.
The main insight here is that we get geometrically decreasing series.
In particular, we claim that $\Trand_H(N,h) \le 4c \cdot \Split(N,\Ho)$, which
can easily be proven by induction:
\begin{align*}
    \Trand_H(N,h) 
     & \le \frac{1}{2} \cdot \left(4c\cdot \Split(\frac{N}{2h},\Ho)  +4c\cdot \Split(\frac{N}{2h},\Ho)\right) \\
     & +   \frac{1}{2} \cdot \left(\sum_{i = 1}^{2h}4c\cdot\Split(\frac{N}{2h},\Ho)\right)   + c \cdot {\Split(N,\Ho)}   \\
     & =   4c(h+1) \cdot \Split(\frac{N}{2h},\Ho)+ c \cdot \Split(N,\Ho) \\
     & \le \frac{4c(h+1)}{2h} \cdot \Split(N,\Ho)+ c \cdot \Split(N,\Ho)  \\
     & =   \left(\frac{4(h+1)}{2h} +1\right) \cdot c \cdot \Split(N,\Ho) \\
     & \le 4c \cdot \Split(N,\Ho) \text{\hspace{6cm} for any } h \ge 2.
\end{align*}
Observe that the same applies to the convex hull algorithm. 
Thus, we have the following theorem.

\begin{restatable}{theorem}{rand-overallIO}\label{thm:rand}
    For a set $P$ of $N$ points in the plane, there exists a 
    randomized cache-oblivious algorithm that finds the maxima of $P$ or
    the convex hull of $P$ in $\OhOf{N \log H}$ worst-case time
    and $\OhOfExp{n \log_m H}$ expected I/Os, 
    where $H$ is the size of the output and $m = M/B$ and $n = N/B$.
\end{restatable}

\bibliographystyle{plainurl}
\bibliography{refs}

\clearpage

\appendix

\crefalias{section}{appendix} 

\ifArxiv
  
\fi

\end{document}